\documentclass[a4paper,fleqn]{cas-dc}
\usepackage{amsmath,amsfonts}
\usepackage{algorithmic}
\usepackage{algorithm}
\usepackage{array}
\usepackage[caption=false,font=normalsize,labelfont=sf,textfont=sf]{subfig}
\usepackage{textcomp}
\usepackage{url}
\usepackage{verbatim}
\usepackage{graphicx}
\usepackage[numbers]{natbib}
\usepackage{tabularx}
\usepackage{xcolor}  
\usepackage{ifthen}
\usepackage{subcaption}
\usepackage{url}
\usepackage{bm}

\newtheorem{assumption}{Assumption}
\newtheorem{theorem}{Theorem}
\newtheorem{proposition}[theorem]{Proposition}
\newtheorem{corollary}[theorem]{Corollary}

\newcommand{\nomentry}[2]{%
  \ensuremath{#1} & #2\\[1pt]
}

\allowdisplaybreaks[3]

\newproof{proof}{Proof}

\begin{document}

\shorttitle{Carbon-Aware Data Center Workload Allocation: Emission Disclosure, Capacity Leasing, and Contract Reshuffling}
\shortauthors{Chen et al.}

\title[mode=title]{Carbon-Aware Data Center Workload Allocation: Emission Disclosure, Capacity Leasing, and Contract Reshuffling}

\author[1,2]{Yihsu Chen}
\ead{yihsuchen@ucsc.edu}

\author[2]{Fargol Nematkhah}
\ead{fnematkh@ucsc.edu}

\author[3]{Abel Souza}
\ead{absouza@ucsc.edu}

\author[4]{Andrew L. Liu}
\cormark[1]
\ead{andrewliu@purdue.edu}

\affiliation[1]{
 organization={Environmental Studies Department, University of California, Santa Cruz},
 city={Santa Cruz},
 state={CA},
 country={USA}}

\affiliation[2]{
 organization={Electrical and Computer Engineering Department, University of California, Santa Cruz},
 city={Santa Cruz},
 state={CA},
 country={USA}}

\affiliation[3]{
 organization={Computer Science and Engineering Department, University of California, Santa Cruz},
 city={Santa Cruz},
 state={CA},
 country={USA}}

\affiliation[4]{
 organization={Edwardson School of Industrial Engineering, Purdue University},
 city={West Lafayette},
 state={IN},
 country={USA}}
\cortext[1]{Corresponding author}




\begin{abstract}
The rapid adoption of AI has driven rapid growth in computational demand, with large language models at the forefront since ChatGPT's debut in 2022.
Meanwhile, large amounts of renewable energy are ultimately curtailed due to transmission congestion and inadequate demand.
This work develops a power market model that allows hyperscalers to spatially migrate LLM inference workloads to geo-distributed modular datacenters (MDCs) co-located with renewable generation at the edge of the network.
We introduce the optimization problems faced by the hyperscaler and MDCs in addition to consumers, producers, and the electric grid operator, where the hyperscaler leases MDC capacity while ensuring that required service level objectives (SLOs) are met.
The overall market model is formulated as a complementarity problem, for which we establish equilibrium existence and uniqueness of certain aggregate market quantities.
We further show that bilateral contract allocations can vary while preserving the same physical market outcome, so cleaner contract-attributed procurement need not imply additional clean generation.
Applying the model to the IEEE RTS-24 bus system, we find that even when MDCs disclose the CO$_2$ emissions associated with their energy supply, renting less polluting MDCs yields limited system emission reductions because of \textit{contract reshuffling}.
This effect can be mitigated if conventional loads are supplied through forward contracts such as power purchase agreements.
Interestingly, this also reduces system congestion as the hyperscaler becomes increasingly cost-aware. 
\end{abstract}

\begin{keywords}
Data centers \sep Carbon-aware computing \sep Modular data centers
\sep Electricity markets \sep Complementarity problems
\sep Contract reshuffling \sep Emission disclosure
\end{keywords}
\maketitle
\begin{table*}[pos=t]
\centering
\setlength{\fboxsep}{8pt}
\fbox{%
\begin{minipage}{0.965\textwidth}

\textbf{Nomenclature}

\vspace{0.7em}

\begin{minipage}[t]{0.485\textwidth}

\textit{Sets and indices}

\vspace{0.3em}

\begin{tabularx}{\linewidth}{@{}l@{\hspace{1em}}X@{}}
\nomentry{B}{Set of all batches}
\nomentry{I^d}{Set of buses hosting only conventional load}
\nomentry{I^\kappa}{Set of buses with a hyperscale datacenter}
\nomentry{I^\chi}{Set of buses with a modular datacenter}
\nomentry{I}{Set of all buses; $I=I^d\cup I^\kappa\cup I^\chi$}
\nomentry{H_j}{Set of all generators connected to bus $j\in I$}
\nomentry{H_i^\omega}{Set of renewable resources connected to bus $i$}
\nomentry{K}{Set of all transmission lines}
\nomentry{K_i}{Set of transmission lines connected to bus $i\in I$}
\nomentry{T}{Set of all time periods}
\end{tabularx}

\vspace{0.7em}
\textit{Parameters}

\vspace{0.3em}

\begin{tabularx}{\linewidth}{@{}l@{\hspace{1em}}X@{}}
\nomentry{\delta}{Weight coefficient in the hyperscale datacenter's objective function}
\nomentry{\nu}{Coefficient representing the number of GPUs consuming one MW in an hour (GPU/MW)}
\nomentry{\mathrm{Cap}_i}{Power capacity of modular datacenter at location $i$ (MW)}
\nomentry{e_{it}}{Carbon emission rate at bus $i$ and time $t$ (ton/MWh)}
\nomentry{F_k}{Thermal limit of transmission line $k$ (MW)}
\nomentry{g^c_{iht}}{Curtailed energy of renewable resource $h$ at bus $i$ at time $t$ (MW)}
\nomentry{G_{ih}}{Capacity of generator $h$ connected to bus $i$ (MW)}
\nomentry{\mathrm{G}_{jhit}^{FC}} {Producers' sales in the absence of data centers workload (MW)}
\end{tabularx}

\end{minipage}
\hfill
\begin{minipage}[t]{0.485\textwidth}


\vspace{0.3em}

\begin{tabularx}{\linewidth}{@{}l@{\hspace{1em}}X@{}}
\nomentry{\mathrm{PTDF}_{ki}}{Power transfer distribution factor of transmission line $k$}
\nomentry{q_b}{Computing load of batch $b$ (MWh)}
\nomentry{\tau}{Producers' forward contract position (unitless)}
\end{tabularx}

\vspace{0.7em}
\textit{Decision variables}

\vspace{0.3em}

\begin{tabularx}{\linewidth}{@{}l@{\hspace{1em}}X@{}}
\nomentry{d_{jhit}}{Conventional electricity demand at location $i$ from power plant $h$ at bus $j$ at time $t$ (MWh)}
\nomentry{g_{jhit}}{Energy output of generator $h$ at bus $j$ delivered to bus $i$ at time $t$ (MWh)}
\nomentry{y_{it}}{Energy injection/withdrawal at bus $i$ at time $t$ (MWh)}
\nomentry{\alpha_{bit}}{GPU leasing price for batch $b$ at MDC $i$ at time $t$ (\$/GPU)}
\nomentry{p_{it}}{Energy withdrawal of modular datacenter at bus $i$ and time $t$ (MWh)}
\nomentry{k^r_{bit}}{Energy consumption of batch $b$ agreed by the modular datacenter at bus $i$ at time $t$ (MWh)}
\nomentry{s_{it}}{Energy spillover of contracted renewable energy at bus $i$ at time $t$ (MWh)}
\nomentry{k^s_{bit}}{Energy consumption of batch $b$ contracted by the hyperscaler to be processed at bus $i$ at time $t$ (MWh)}
\nomentry{\ell_{bjhit}}{Energy purchased from plant $h$ at bus $j$ to process batch $b$ locally at hyperscaler $i$ at time $t$ (MWh)}
\nomentry{\theta^a_{jhit}}{Bilateral contract price for participant type $a\in\{d,\chi,\kappa\}$ (\$/MWh)}
\end{tabularx}
\end{minipage}
\end{minipage}%
}
\end{table*}
\section{Introduction}\label{sec:intro}

As AI-driven datacenter loads become large enough to materially affect
power systems, an important question is no longer only where a datacenter
consumes electricity, but what changes elsewhere in the market when it
seeks cleaner electricity. Consider a hyperscaler that increases its
procurement from a low-emission generator. If not all electricity
procurement is locked in by pre-existing contracts, another consumer may
reduce its purchases from that generator and instead procure electricity
from a higher-emission generator. The hyperscaler's contract-attributed
emissions can therefore fall even though the output of each generator,
and hence total system emissions, remains unchanged. We refer to this
reallocation of electricity purchases across bilateral contracts as
\textit{contract reshuffling}.

This issue is becoming increasingly important as the rapid adoption of
AI, particularly large language models (LLMs), drives substantial growth
in datacenter electricity demand. At the same time, major technology firms are making large, long-term commitments to clean electricity procurement, such as Microsoft's agreement with Constellation supporting the restart of the Crane Clean Energy Center and Google's agreement with Kairos Power \cite{constellation2024,terrell25}. As both the scale of datacenter demand and the volume of clean procurement grow, whether such procurement translates into actual reductions in system-wide emissions becomes correspondingly consequential.

A substantial literature studies carbon-aware datacenter operation by
shifting computing workloads across time or location in response to
electricity prices, renewable availability, or carbon-intensity signals
\cite{dou17,dandres17,lindberg21,paramanayakam25,stojkovic25}.
Much of this literature treats these signals as exogenous to the
datacenter's workload decisions. This can be problematic for large
workload shifts: marginal emission rates characterize the response to an
incremental change in load and may not accurately predict the resulting
emissions when a nontrivial load shift changes generation dispatch and
network congestion. Recent studies have therefore begun to account for
the endogenous response of the power system. Jiang et al.~\cite{jiang26},
for example, formulate an equilibrium model of carbon-sensitive
electricity consumers and show that average carbon-intensity signals
need not induce the intended system-wide emission reductions. Other
recent studies similarly demonstrate that the emissions consequences of
datacenter flexibility depend on the resulting generation response and
power-system conditions \cite{riepin25,senga26}. These studies establish
the importance of modeling the physical power-system response, but do
not examine how bilateral electricity procurement and the contractual
attribution of generation adjust as datacenter workloads are reallocated.

A separate electricity-market literature does examine how contractual
electricity procurement can diverge from physical system outcomes. In
the context of regional carbon regulation, \cite{chen11} shows that
bilateral electricity contracts can be reshuffled so that cleaner
generation is attributed to regulated consumers while more
carbon-intensive generation is reassigned elsewhere, without a
corresponding reduction in total emissions. Related work examines
closely related behavior under the term \textit{resource reshuffling}
\cite{bushnell14,fowlie21}. This literature, however, does not consider
carbon-aware computing, workload migration, or endogenous transactions
in computing capacity. The interaction between carbon-aware workload
allocation and electricity-contract reshuffling therefore remains
largely unexplored.

This interaction is becoming particularly relevant with the emergence
of modular datacenters (MDCs). MDCs are prefabricated, factory-integrated facilities
that package computing, power, and cooling infrastructure into
standardized units that can be deployed and expanded more rapidly than
conventional site-built datacenters. Commercial systems are already
available for high-density AI and hyperscale applications
\cite{schneider25}, and Jones Lang LaSalle (JLL) projects annual global
sales of MDCs to increase from approximately \$11 billion in 2025 to
\$48 billion by 2030 \cite{jll26}.

This deployment flexibility creates a potentially attractive pathway for
carbon-aware computing. MDCs can be located near renewable generation
that would otherwise be curtailed, allowing otherwise-wasted clean
energy to serve geographically flexible inference workloads. This
benefit, however, is limited by the amount of curtailed renewable energy
available. Once an MDC's electricity demand exceeds this supply, it must
procure additional electricity from the power market. The system-wide
emission effect therefore depends not only on the curtailed renewable
energy utilized by the MDC, but also on how its additional electricity
demand changes generation and electricity procurement elsewhere in the
market.

This interaction motivates the central question of this paper:
can carbon-aware migration of inference workloads to MDCs that utilize
otherwise-curtailed renewable energy deliver genuine system-wide emission
reductions?

To answer this question, we consider a hyperscaler that can process
inference workloads locally or lease computing capacity from geographically
distributed MDCs co-located with renewable generation that would otherwise
be curtailed. We jointly model capacity leasing between the hyperscaler and
MDCs, bilateral electricity procurement among generators and loads, and
network-constrained power-market equilibrium. We consider two
emission disclosure schemes. Under the \textit{ex post} scheme, MDC
emissions are disclosed after operation. Under the \textit{ex ante}
scheme, MDCs disclose their emission intensities before workload
allocation, allowing the hyperscaler to incorporate this information into
its capacity-leasing offers and workload decisions. This framework
determines both contract-attributed datacenter emissions and physical
system emissions within the same market equilibrium.

The contributions of this paper are threefold. First, we formulate the
interactions among conventional consumers, producers, the grid operator,
the hyperscaler, and MDCs as a complementarity problem, providing an
integrated framework for carbon-aware workload allocation, computing
capacity leasing, and bilateral electricity procurement. 

Second, we establish existence of a market equilibrium and uniqueness of
aggregate consumer demand and generator output under affine marginal
benefit and cost functions. We then characterize contract reshuffling by
showing that bilateral electricity allocations can vary while preserving
buyer-side procurement and generator output. Consequently,
contract-attributed emissions can decline without a corresponding change
in physical system emissions. 

Third, numerical results from an illustrative IEEE RTS-24 case study
demonstrate the practical implications of these theoretical findings.
Under the \textit{ex ante} scheme, workload-attributed emissions can
decline by about 17\% relative to the \textit{ex post} case, while the
corresponding reduction in total system emissions is less than 0.1\%.
This occurs even though the hyperscaler responds to carbon information
by differentiating its capacity-leasing offers across MDCs. When
conventional consumers retain larger portions of their bilateral
purchases through forward contracts, the scope for contract reshuffling
is reduced and larger physical emission reductions emerge. A notable network-level benefit also emerges under sufficiently
restrictive forward-contract positions. As the hyperscaler becomes more
cost-aware, it shifts additional workload to MDCs at the network edge,
which reduces system congestion. This effect is much less pronounced at
looser contract positions, where greater scope for reshuffling remains.

The remainder of the paper is organized as follows.
Section~\ref{sec:model} presents the optimization problems and market
equilibrium formulation. Section~\ref{sec:theory} establishes the
theoretical properties of the model and characterizes contract
reshuffling. Section~\ref{sec:case} presents the numerical case study,
and Section~\ref{sec:concl} concludes.

\section{Model} \label{sec:model}
This section presents the optimization problem faced by each entity.
The network consists of node $i \in I$. 
The nodes hosting conventional loads, the hyperscalers, and the MDCs are, respectively denoted by $I^d$, $I^\kappa$, and $I^\chi$, where the three sets are mutually exclusive. That is, $I^d \cap I^\kappa = I^d \cap I^\chi = I^\kappa \cap I^\chi = \emptyset$.

\subsection{Conventional Consumer’s Problem}
The consumers' willingness to pay in node $i$ at period $t$ is represented by an affine benefit function $B_{it}$.
The utility in node $i$ procures electricity on behalf of its customers in node $i$ by entering a power purchase agreement (PPA) with producers.
Its objective is given in (\ref{eq:consumer_obj}) where $d_{jhit}$ denotes the energy sales from the power plant $h$ located in $j$ to $i$, and where the first and second terms represent the total benefit and the payments, respectively. 
\begin{align}
\max_{d_{jhit}\ge 0} \quad & \sum_t B_i\big(\sum_{j,h \in H_j} d_{jhit}\big) - \sum_{j,h \in H_j,t} \theta^d_{jhit} d_{jhit} \label{eq:consumer_obj}
\end{align}

\subsection{Producer’s Problem}
Each generator $h\in H_j$ located at bus $j\in I$ solves the profit-maximization problem in \eqref{eq:producer_obj}--\eqref{eq:producer_cons1}. The decision variable $g_{jhit}$ denotes the quantity of electricity sold by generator $h$ at bus $j$ to buyer $i$ at time $t$, subject to the generator capacity constraint in \eqref{eq:producer_cons1}.
 Variables $\theta_{jhit}^d$, $\theta_{jhit}^\chi$ and $\theta_{jhit}^\kappa$ denote the bilateral prices with utilities, the hyperscalers and MDCs, respectively. 
Note that their values are exogenous to the producers' problem but endogenous to the overall market equilibrium model.
The first three summations in \eqref{eq:producer_obj} represent revenues from electricity sales to MDCs, hyperscalers, and conventional consumers, respectively. The producer's profit is given by these revenues minus the production cost $C_{jh}$ and the wheeling charges, represented by the nodal charge difference $(\omega_{it}-\omega_{jt})$ associated with delivering electricity from bus $j$ to bus $i$.\footnote{For a detailed discussion on wheeling charges, please refer to \cite{hobbs01}.}
\begingroup
\setlength{\mathindent}{3pt}
\begin{align}
\max_{\{g_{jhit}\geq 0\}_{i\in I,\,t\in T}} \quad
& \sum_{t\in T}\Bigg[
    \sum_{i\in I^\chi} g_{jhit}\theta^\chi_{jhit}
    + \sum_{i\in I^\kappa} g_{jhit}\theta^\kappa_{jhit}
    \nonumber\\
&\qquad
    + \sum_{i\in I^d} g_{jhit}\theta^d_{jhit}
    - C_{jh}\!\left(\sum_{i\in I}g_{jhit}\right)
    \nonumber\\
&\qquad
    - \sum_{i\in I}
    (\omega_{it}-\omega_{jt})g_{jhit}
    \Bigg]
\label{eq:producer_obj}\\
\text{s.t.}\quad
& \sum_{i\in I}g_{jhit}\leq G_{jh}
\quad(\lambda_{jht}),
\qquad \forall t\in T .
\label{eq:producer_cons1}
\end{align}
\endgroup

\subsection{The Grid Operator}
The grid operator aims to maximize the value of the transmission network by determining the net injection or withdrawal $y_{it}$ where $\omega_{it}$ gives the wheeling charge to bring power from the hub to node $i$ at time $t$.  
The grid operator is subject to the energy balance conditions in (\ref{eq:iso-energy}) 
and to the lower and upper transmission flow limits in 
(\ref{eq:iso-flow1}) and (\ref{eq:iso-flow2}), respectively, where PTDF 
refers to the Power Transfer Distribution Factors derived from the 
linearized DC power flow model~\cite{schweppe88}.
A similar formulation has been used elsewhere, e.g., Metzler et al. \cite{metzler03}.

\begin{align}
\hspace*{-150pt}
\max_{y_{it}\ \mathrm{free}} \quad & \sum_{i,t} \omega_{it} y_{it} \\
\text{s.t.}\quad & \sum_i y_{it}=0 \quad (\gamma_t) \quad \forall t\in T \label{eq:iso-energy}\\
& -F_k - \sum_{i \in I} \text{PTDF}_{ki} y_{it} \le 0 \ (\mu^1_{kt}), \ \forall k\in K,\ t\in T\label{eq:iso-flow1}\\
& \sum_{i \in I} \text{PTDF}_{ki} y_{it} - F_k \le 0 \ (\mu^2_{kt}),  \ \forall k\in K,\ t\in T \label{eq:iso-flow2}.
\end{align}

\subsection{Modular datacenter’s Problem}

MDCs are assumed to be strategically co-located at the edge of the network where curtailment constantly occurs. 
The quantities of expected ``curtailed'' renewables are denoted by $\sum_{h \in H^w_i} g^c_{iht}$. 
An MDC $i\in I^\chi$ maximizes its profits by deciding on power purchase contracts  $p_{jhit}$, leasing contracts with hyperscalers $k^{r}_{bit}$, and spillover quantities $s_{it}$.
\textcolor{black}{Note that, as alluded to earlier, not every MDC can process all the batches; in fact, each MDC has a different latency threshold, which can be approximated by the distance between the hyperscaler and MDC \cite{{goonatilake12}}. 
}

The set of batches that can be processed by the MDCs located in node $i$ is denoted by $b \in B_i$, which is determined by the distance between the hyperscaler and the MDCs as well as processing time, which depends on the complexity of the requests characterized by the input and output tokens \cite{wilkins25}.
The set of variables controlled by the MDCs collectively is denoted as $\bm{\Phi}=\{p_{jhit}, k^r_{bit}, s_{it} \}$.

The maximization problem faced by the MDCs is in (\ref{eq:mdc_obj})--(\ref{eq:mdc_cons3}).  
The first term $\alpha_{bit}\, k^r_{bit}\, \nu$ in the objective function \eqref{eq:mdc_obj} describes the revenue from leasing GPUs to the hyperscaler, where $\alpha_{bit}$ is the revenue from leasing per unit of GPUs, $k^r_{bit}$ denotes the leasing capacity, and $\nu$ denotes the power rating per GPU.
$I^\chi$ represents the set of nodes with MDCs.
The second term represents the power purchase costs, where $p_{jhit}$ and 
$\theta^\chi_{iht}$ denote the quantity and price of the power purchase agreement with suppliers, respectively.

\begingroup
\setlength{\mathindent}{3pt}
\begin{align}
\max_{\bm{\Phi}_i \geq 0}\quad 
& \sum_{b,t} \alpha_{bit}\, k^r_{bit}\, \nu
- \sum_{j\in I}\sum_{h\in H_j}\sum_{t\in T}
p_{jhit}\,\theta^\chi_{jhit}
\label{eq:mdc_obj} \\[0.5em]
\text{s.t.}\quad 
& \sum_b k^r_{bit}
-\sum_{j\in I}\sum_{h\in H_j}p_{jhit}
+s_{it}
-\sum_{h\in H_i^\omega}g^c_{iht}
=0
\quad(\eta_{it}),
\qquad \forall t\in T,
\label{eq:mdc_cons1}\\[0.5em]
& \sum_b k^r_{bit}-\mathrm{Cap}_i
\leq 0
\quad(\rho_{it}),
\qquad \forall t\in T,
\label{eq:mdc_cons2}\\[0.5em]
& s_{it}-\sum_{h\in H_i}g^c_{iht}
\leq 0
\quad(\upsilon_{it}),
\qquad \forall t\in T.
\label{eq:mdc_cons3}
\end{align}
\endgroup
The problem is subject to three constraints. 
Equation~(\ref{eq:mdc_cons1}) imposes energy balance. 
Equation~(\ref{eq:mdc_cons2}) limits the energy consumed by the migrated 
batches to be less than or equal to the capacity of the MDCs. 
Equation~(\ref{eq:mdc_cons3}) requires that the spillover $s_{it}$ be less  than or equal to the amount of curtailed energy.
We next turn to the hyperscaler's problem.

\subsection{Hyperscale datacenter’s Problem}
The hyperscale datacenter’s problem is given in \eqref{eq:hyp_obj}--\eqref{eq:hyp_cons1}. The hyperscaler receives inference requests, aggregates them into batches $q_b$ for $b \in B$, and decides whether to process them locally using $\ell_{bjhit}$ or to send them to MDCs using $k^{s}_{bit}$, with decision variables collected in $\Theta := \{\ell_{bjhit}, k^{s}_{bit}\}$. The hyperscaler accounts for both processing costs and associated CO$_2$ emissions, as reflected in the objective \eqref{eq:hyp_obj}, where the term weighted by $\delta$ captures processing costs and the term weighted by $1-\delta$ captures CO$_2$ emissions. Specifically, $\sum_{b,i \in I^\kappa,t} \alpha_{bit} k^{s}_{bit} \nu$ and $\sum_{b,j,h \in H_j, i \in I^\kappa,t} \ell_{bjhit} \theta^\kappa_{ijht}$ represent local processing and external leasing costs, while $\sum_{j,h \in H_j, i \in I^\chi,t} p_{jhit} e_{jh}$ and $\sum_{b,j,h \in H_j, i \in I^\kappa,t} \ell_{bjhit} e_{jh}$ represent local and external CO$_2$ emissions, respectively. This formulation corresponds to an {\it ex post} emission disclosure regime for MDCs. The constraint in \eqref{eq:hyp_cons1} ensures that all requests are processed.
\begingroup
\setlength{\mathindent}{3pt}
\begin{align} 
\min_{\bm{\Theta} \geq 0}\quad & \delta\!\left(\sum_{b,i \in I^\chi,t}\alpha_{bit}k^s_{bit}\nu + \sum_{b,j,h \in H_j, i \in I^{\kappa},t} \ell_{bjhit}\, \theta^\kappa_{jhit}\right) \nonumber\\ 
 + (1-\delta) & \left(\sum_{b,j,h \in H_j, i \in I^{\kappa},t} \ell_{bjhit}\, e_{jh}+\hspace*{-5pt}\sum_{j,h \in H_j, i' \in I^\chi, t} p_{jhi't}\, e_{jh} \right) \label{eq:hyp_obj}\\ 
\text{s.t.}\quad & \sum_{i \in I^{\chi}} k^s_{bit} + \sum_{j \in I,\, h \in H_j,\, i \in I^\kappa,t} \ell_{bjhit} - q_b = 0 \quad (\psi_{bt}) \nonumber\\
& \quad  \forall b \in B,\, t \in T \label{eq:hyp_cons1}
\end{align}
\endgroup

\subsection{Market Clearing Conditions}
The following five sets of market clearing conditions, (\ref{eq:mk_1})--(\ref{eq:mk_5}) are included for all $j\in I$, $h \in H_j$, and $t \in T$. 
The first three conditions define the bilateral contract prices 
between suppliers and (i) consumers, (ii) hyperscalers, and (iii) MDCs. 
The fourth condition specifies the leasing prices between hyperscalers and MDCs, 
while the fifth condition determines the wheeling charge. \\[-6pt]
\begingroup
\setlength{\mathindent}{2em}
\begin{align}
&\theta^d_{jhit}\ \text{free},\hspace{0.1cm} d_{jhit} - g_{jhit} = 0, \ \forall  i \in I^d\label{eq:mk_1}\\
& \theta^\chi_{jhit}\ \text{free},\hspace{0.1cm} p_{jhit} - g_{jhit} = 0, \ \forall  i \in I^\chi \\
& \theta^\kappa_{jhit}\ \text{free},\hspace{0.1cm}\sum_b \ell_{bjhit} - g_{jhit} = 0, \ 
 \forall  i \in I^\kappa\label{eq:hyper_MC} \\
& \omega_{it}\ \text{free},\hspace{0.1cm}
y_{it}
=
\sum_{j\in I}\sum_{h\in H_j} g_{jhit}
-
\sum_{h\in H_i}\sum_{j\in I} g_{ihjt},
\quad
\forall i\in I
\label{eq:mk_5}\\
& \alpha_{bit}\ \text{free},  \hspace{0.1cm} k^s_{bit}-k^r_{bit} = 0, \hspace{0.1cm} \forall b \in B, i \in I^\kappa. \label{eq:mk_alpha}
\end{align}
\endgroup

\subsection{Equilibrium Models}
Given that the problems faced by all the entities are convex, a necessary and sufficient condition for an optimal solution to satisfy is given by the following first-order or Karush-Kuhn-Tucker (KKT) conditions, one for each variable. 

\noindent {\it{Consumers}}
\begin{multline}
\hspace*{-20pt}
0 \le d_{jhit} \perp\ -B'_i\Big(\sum_{j,h \in H_j} d_{jhit}\Big) + \theta^d_{jhit} \ge 0,\\
\forall j \in I,\ h \in H_j,\ i \in I^d,\ t \in T.
\label{eq:kkt_1}
\end{multline}

\noindent {\it{Producers}}
\begingroup
\setlength{\mathindent}{4pt}
\begin{align}
0 \le g_{jhit} & \perp\ -\sum_{m :\{\chi, \kappa, d\} }\theta^m_{jhit}+(\omega_{it}-\omega_{jt}) +C'_{jh}(\sum_{i'\in I}g_{jhi't}) \nonumber \\
& \quad +\lambda_{iht}\ge 0,  \quad \forall j \in I, h \in H_j, i \in I, t \in T \label{eq:kkt_g}\\
0 \le \lambda_{jht} & \perp\ G_{j,h} - \sum_{i\in I} g_{jhit} \ge 0, \quad \forall j \in I, h \in H_j, t \in T. \label{eq:kkt_lambda}
\end{align}
\endgroup

\noindent {\it{The Grid Operator}}\\
For $\forall i \in I, t \in T$:
\begin{align}
\hspace*{-10pt}
& y_{it} \ \text{free,} && \omega_{it} + \sum_{k \in K} \text{PTDF}_{ki}(\mu^1_{kt}-\mu^2_{kt}) + \gamma_t = 0 \label{eq:kkt_yfree}
\\[0.25em]
& \gamma_{t} \ \text{free,} && \sum_{i\in I} y_{it} = 0 \label{eq:sumy_0}
\end{align}
For $\forall k \in K, t \in T$:
\begin{align}
& 0 \leq \mu^1_{kt} & \perp & -F_k - \sum_{i \in I} \text{PTDF}_{ki} y_{it} \le 0 \label{eq:kkt_TCap}
\\[0.25em]
& 0 \leq \mu^2 _{kt} & \perp & \sum_{i \in I} \text{PTDF}_{ki} y_{it} - F_k \le 0
\label{eq:kkt_iso}
\end{align}

\noindent {\it{Modular datacenter}}
\begin{align}
& 0 \leq p_{jhit} 
  && \perp && -\theta^\chi_{jhit} - \eta_{it} \leq 0, \nonumber\\
  & && && \quad \forall j \in I, h \in H_j, i \in I, t \in T \label{eq:Modular_p}\\[0.25em]
& 0 \leq k^r_{bit} 
  && \perp && \alpha_{bit}\nu + \eta_{it} - \rho_{it} \leq 0 \nonumber \\
  & && && \quad \forall b \in B, i \in I, t \in T.\label{eq:Modular_k}
\end{align}

For $\forall i \in I, t \in T$:
\begin{multline}
\eta_{it} \ \text{free}, \quad
\sum_{b\in B} k^r_{bit} - \sum_{j \in I, h\in H_j} p_{jhit} 
       + s_{it} - \sum_{h\in H_i^\omega} g^c_{iht} = 0  \label{eq:mdc_energy_balance}
\end{multline}       

\begin{align}       
& 0 \leq \rho_{it} 
  && \perp && \sum_{b\in B} k^r_{bit} - \mathrm{Cap}_i \leq 0 \label{eq:kkt_modular_cap}\\[1em]
& 0 \leq s_{it} 
  && \perp && \eta_{it}-\upsilon_{it}  \leq 0 \label{eq:kkt_modular_etav}\\ 
& 0 \leq \upsilon_{it} 
  && \perp && s_{it} - \sum_{h\in H_i} g^c_{iht} \leq 0 \label{eq:kkt_s_gc}
\end{align}

\noindent {\it{Hyperscalers}}
\begin{align}       
 & 0 \leq \kappa^s_{bit} 
   && \perp && -\delta \nu \alpha_{bit} - \psi_{bt}  \leq 0, \nonumber\\
  & && &&  \forall b \in B, i \in I^\kappa, t \in T  \label{eq:kkt_hyp1}
\end{align}
\vspace{-0.5em}
\begin{align}       
& 0 \leq \ell_{bjhit} 
  && \perp && -\delta\theta^\kappa_{jhit}-(1-\delta)e_{jh}- \psi_{bt}\leq 0 \nonumber \\
 & && && \forall b \in B, i \in I^\kappa, t \in T \label{eq:kkt_l}
\end{align}

\begin{align}  
& \psi_{bt}\ \text{free}, \quad \sum_{i \in I^{\kappa}} k^s_{bit} + \sum_{j \in I,\, h \in H_j,\, i \in I^\kappa} \ell_{bjhit} - q_b = 0  \nonumber\\
& \quad \hspace{5cm} \forall b \in B,\, t \in T \label{eq:kkt_N}
\end{align}
The market equilibrium problem is then defined as the collection of all KKT conditions~(\ref{eq:kkt_1})--(\ref{eq:kkt_N}), together with the market clearing conditions~(\ref{eq:mk_1})--(\ref{eq:mk_alpha}).
The resulting problem can then be solved using complementarity problem solvers such as PATH~\cite{path00}, Knitro~\cite{knitro25}, and Gurobi~\cite{gurobi25}.

\subsection{{\it{Ex Ante} Emission Disclosure}}
Alternatively, a hyperscaler may require each MDC to disclose its procurement mix and associated emissions before the hyperscaler allocates workload.
The benefit of this arrangement is that a hyperscaler can directly reduce
the emissions attributed to its workload through the capacity-leasing
agreement.
Under this situation, a hyperscaler tends to provide a more competitive leasing offer to those MDCs with a lower emissions intensity,  leading to a higher $\alpha_{bit}$.

To study this situation, we modify the aforementioned model by adding a new market clearing condition that calculates the MDCs' emission intensity as follows:
\begin{align}       
& \epsilon_{it} = \frac{\sum_{j \in I, h \in H_j} p_{jhit}e_{jh}}{\sum_{b \in B} k^{r}_{bit}}, \quad \forall i \in I^x, t \in T. \label{eq:MDC_em}
\end{align}
We then replace the first term within the parenthesis preceded by 
$(1-\delta)$ in (\ref{eq:hyp_obj}) with 
$\sum_{j \in I^\chi} \epsilon_{it} \kappa^{s}_{bit}$.
This yields the following condition for $k^\kappa_{bit}$.
\begin{align}       
 & 0 \leq \kappa^s_{bit} 
   && \perp && - \delta\nu\alpha_{bit} - (1-\delta)\epsilon_{it}-\psi_{bt}  \leq 0, \nonumber\\
  & && &&  \forall b \in B, i \in I^\kappa, t \in T \label{eq:hy_rate}
\end{align}
Note that the term $(1-\delta)\epsilon_{it}$ in \eqref{eq:hy_rate}
makes the optimality condition location-dependent.
As $\epsilon_{it}$ differs across $i$, the equilibrium leasing price
$\alpha_{bit}$ can vary across MDCs, unlike in \eqref{eq:kkt_hyp1}.
We illustrate this observation in the numerical case study in Section \ref{sec:case} below.

\section{Theoretical Analysis}\label{sec:theory}
This section reformulates the equilibrium conditions as a single mixed linear complementarity problem (MLCP), establishes existence, and proves uniqueness of certain aggregate quantities using monotonicity of the associated operator.

\subsection{MLCP Formulation}
\label{subsec:affine}
To express the MLCP in a compact matrix form, we require explicit functional forms for the consumer benefit functions $B_i(\cdot)$ in \eqref{eq:consumer_obj} and the power producer cost functions $C_{jh}(\cdot)$ in \eqref{eq:producer_cons1}.
We hence impose the following standing assumption: 
\begin{assumption}[Affine marginal benefit and marginal cost]
\label{ass:affine}
For each $i \in I^d$ and $t \in T$, the marginal benefit is affine in the total consumption,
\begin{equation}\label{eq:Marginal_Benefit}
B'_i\!\left(\sum_{j \in I,\, h \in H_j} d_{jhit}\right)
= b^0_{it} - \sum_{j \in I,\, h \in H_j} b^1_{it}\, d_{jhit},
\qquad b^1_{it} \ge 0.
\end{equation}
For each $j\in I$, $h\in H_j$, and $t\in T$, the marginal
production cost is affine in the producer's total output,
\begin{equation}
\label{eq:Marginal_Cost}
C'_{jh}\!\left(\sum_{i\in I} g_{jhit}\right)
=
c^0_{jh}
+
c^1_{jh}\sum_{i\in I} g_{jhit},
\qquad
c^1_{jh}\ge 0.
\end{equation}
\end{assumption}
An affine marginal benefit function with $b^1_{it} \ge 0$ implies a concave quadratic total benefit function and, equivalently, a downward sloping inverse demand function. Empirical evidence suggests that aggregate electricity demand response to price fluctuations is well-approximated by linear specifications over relevant operating ranges \cite{DRCurve}.

On the supply side, an affine marginal cost (with $c^1_{jh} \ge 0$) implies a convex quadratic total cost function. In power system modeling, the heat rate curve of thermal generators, including coal, gas, and oil units, is commonly represented as a quadratic function of output, which directly yields a linear marginal cost. This modeling practice is well documented in classic power system references such as \cite{WoodWollenberg}.

We next define the decision vectors used in this formulation. Let the nonnegative variable vector be defined as
\begin{equation}
\label{eq:z_def}
z := \big( d,\ g,\ \lambda,\ \mu^1,\ \mu^2,\ p,\ k^r,\ s,\ \rho,\ \upsilon,\ k^s,\ \ell \big)
\in \mathbb{R}^{n_z}_{+},
\end{equation}
where each block stacks the corresponding components
$\{d_{jhit}\}$,
$\{g_{jhit}\}$,
$\{\lambda_{jht}\}$,
$\{\mu^1_{kt}\}$,
$\{\mu^2_{kt}\}$,
$\{p_{jhit}\}$,
$\{k^r_{bit}\}$,
$\{s_{it}\}$,
$\{\rho_{it}\}$,
$\{\upsilon_{it}\}$,
$\{k^s_{bit}\}$,
and
$\{\ell_{bjhit}\}$,
over their respective index sets defined in
Section~\ref{sec:model}.

Similarly, define the vector of free variables as
\begin{equation}
\label{eq:pi_def}
\pi := \big( \theta^d,\ \theta^\chi,\ \theta^\kappa,\ \omega,\ \alpha,\ y,\ \gamma,\ \eta,\ \psi \big)
\in \mathbb{R}^{n_\pi},
\end{equation}
stacking
$\{\theta^d_{jhit}\}_{i\in I^d}$,
$\{\theta^\chi_{jhit}\}_{i\in I^\chi}$,
$\{\theta^\kappa_{jhit}\}_{i\in I^\kappa}$,
$\{\omega_{it}\}$,
$\{\alpha_{bit}\}$,
$\{y_{it}\}$,
$\{\gamma_t\}$,
$\{\eta_{it}\}$,
and
$\{\psi_{bt}\}$.

With the above notation and Assumption~\ref{ass:affine}, all
complementarity and equality conditions in
\eqref{eq:kkt_1}--\eqref{eq:kkt_N} can be stacked into the following
MLCP:
\begin{subequations}
\label{eq:micp}
\begin{align}
0 \le z
&\perp Mz+N\pi+q \ge 0,
\label{eq:micp_comp}\\
& \quad \ Pz+Q\pi
=r,
\label{eq:micp_eq}
\end{align}
\end{subequations}
where
$M\in\mathbb{R}^{n_z\times n_z}$,
$N\in\mathbb{R}^{n_z\times n_\pi}$,
$P\in\mathbb{R}^{n_\pi\times n_z}$,
$Q\in\mathbb{R}^{n_\pi\times n_\pi}$,
$q\in\mathbb{R}^{n_z}$, and
$r\in\mathbb{R}^{n_\pi}$.

To present the details of the matrices and vectors, we first define
aggregate demand and generation quantities. For each $i\in I^d$ and
$t\in T$, let
\begin{equation}\label{eq:aggD}
D_{it}:=\sum_{j\in I}\sum_{h\in H_j} d_{jhit}
\end{equation}
denote the aggregate demand at bus $i$ and time $t$. Let $D$ denote
the vector stacking $D_{it}$ over all $i\in I^d$ and $t\in T$, and let
$b^1$ denote the corresponding vector stacking $b^1_{it}$. The
slope-dependent terms in the marginal benefit functions can then be
written compactly as the Hadamard product $b^1\circ D$. Further let
$
n_G:=\sum_{j\in I}|H_j|
$
denote the total number of generators. We use
$\operatorname{diag}(A_1,\ldots,A_n)$ to denote the block-diagonal
matrix with matrix blocks $A_1,\ldots,A_n$. Define
\begin{equation}
\label{eq:Hd}
H_d :=
\operatorname{diag}\left(
\left\{
b^1_{it}\mathbf{1}_{n_G}\mathbf{1}_{n_G}^{\top}
\right\}_{i\in I^d,\;t\in T}
\right).
\end{equation}
For each $(i,t)$, the corresponding block maps the bilateral demand
quantities $\{d_{jhit}\}_{j\in I,h\in H_j}$ into the common marginal
benefit term $b^1_{it}D_{it}$. Since $b^1_{it}\geq 0$, each block matrix is
positive semidefinite (PSD), and hence $H_d$ is PSD.

Similarly, for each $j\in I$, $h\in H_j$, and $t\in T$, define the
aggregate output of generator $h$ at bus $j$ as
\begin{equation}\label{eq:aggG}
G_{jht}:=\sum_{i\in I} g_{jhit}.
\end{equation}
Let $G$ denote the vector stacking $G_{jht}$ over all $j\in I$,
$h\in H_j$, and $t\in T$, and let $c^1$ denote the corresponding
vector of marginal-cost coefficients $c^1_{jh}$, replicated over
$t\in T$. 
Define
\begin{equation}
\label{eq:Hg}
H_g :=
\operatorname{diag}\left(
\left\{
c^1_{jh}\mathbf{1}_{|I|}\mathbf{1}_{|I|}^{\top}
\right\}_{j\in I,\;h\in H_j,\;t\in T}
\right).
\end{equation}
For each $(j,h,t)$, the corresponding block maps the bilateral
generation quantities $\{g_{jhit}\}_{i\in I}$ into the common marginal
cost term $c^1_{jh}G_{jht}$. Since $c^1_{jh}\geq 0$, each block is
PSD, and hence so is $H_g$.

Next, define the matrix
$A_\lambda \in
\mathbb{R}^{|\mathcal{I}_\lambda|\times|\mathcal{I}_g|}$,
where
$
\mathcal{I}_\lambda
:=
\{(j,h,t):j\in I,\ h\in H_j,\ t\in T\}
$
and
$
\mathcal{I}_g
:=
\{(j,h,i,t):j\in I,\ h\in H_j,\ i\in I,\ t\in T\}.
$ 
For any row index $(j,h,t)\in\mathcal{I}_\lambda$ and column
index $(j',h',i',t')\in\mathcal{I}_g$, define
\begin{equation}
\label{eq:Alambda}
\big(A_\lambda\big)_{(j,h,t),(j',h',i',t')}
=
\begin{cases}
1, & \text{if } j=j',\ h=h',\ t=t',\\
0, & \text{otherwise}.
\end{cases}
\end{equation}
Thus, $A_\lambda$ aggregates the bilateral sales of each generator
across all buyers, so that
$
\big(A_\lambda g\big)_{jht}
=
\sum_{i\in I} g_{jhit}.
$
The blocks of $M$ associated with $d$, $g$, and $\lambda$ are then 
\begin{equation}
\label{eq:M_dglambda}
\hspace*{50pt}
\begin{bmatrix}
H_d & 0 & 0\\
0 & H_g & A_\lambda^\top\\
0 & -A_\lambda & 0
\end{bmatrix}.
\end{equation}

The remaining nonzero blocks of $M$ arise from complementary
primal--dual pairs in the MDC problem. In particular, let $A_\rho$
aggregate $k^r_{bit}$ over batches $b$ for each $(i,t)$, and let
$A_\upsilon$ map $s_{it}$ to its corresponding spillover constraint.
The additional nonzero blocks are
\[
M_{k^r,\rho}=A_\rho^\top,\qquad
M_{\rho,k^r}=-A_\rho,
\]
and
\[
M_{s,\upsilon}=A_\upsilon^\top,\qquad
M_{\upsilon,s}=-A_\upsilon.
\]
All other blocks of $M$ are zero.

For the remaining terms in \eqref{eq:micp_comp} and \eqref{eq:micp_eq}, the matrices $N$, $P$, and $Q$ collect the coefficients associated with the interactions between the nonnegative variables $z$ and the free variables $\pi$, while $q$ and $r$ collect the corresponding constant terms. Their detailed constructions are provided in Appendix~\ref{sec:AppA}.

\subsection{MLCP Solution Existence}
\label{subsec:exist}

A pair $(z^\star,\pi^\star)$ satisfying \eqref{eq:micp_comp} and
\eqref{eq:micp_eq} is called a solution of the MLCP, which corresponds to a market equilibrium. 
To facilitate the equilibrium existence proof, we construct an equivalent MLCP formulation.

Since $y$ is unrestricted in sign, we write
$
y=y^+-y^-$, with $y^+,\ y^-\geq0$.
For each $i\in I$ and $t\in T$, define
\[
\zeta_{it}
:=
\omega_{it}
+\sum_{k\in K}\mathrm{PTDF}_{ki}
\left(\mu^1_{kt}-\mu^2_{kt}\right)
+\gamma_t.
\]
The equality constraint in \eqref{eq:kkt_yfree} can then be equivalently represented by
\begin{equation}
\label{eq:y_split}
0\leq y^+_{it}\perp - \zeta_{it}\geq0,
\qquad
0\leq y^-_{it}\perp \zeta_{it}\geq0.
\end{equation}

Furthermore, since $\delta>0$, the complementarity residuals in
\eqref{eq:kkt_hyp1} and \eqref{eq:kkt_l} can be divided by $\delta$
without changing their solution sets. Together with the
market-clearing convention in \eqref{eq:mk_5}, and scaling equality
conditions by nonzero constants where necessary, these transformations
yield an equivalent MLCP of the form
\begin{subequations}
\label{eq:micp_hat}
\begin{align}
0\leq\widehat z
&\perp
\widehat M\widehat z
+\widehat N\widehat\pi
+\widehat q
\geq0,
\label{eq:micp_hat_comp}\\
& \quad 
-\widehat N^\top\widehat z+\widehat r
\ =\ 0.
\label{eq:micp_hat_eq}
\end{align}
\end{subequations}
Here, $\widehat z$ is obtained by augmenting $z$ with $y^+$ and
$y^-$, while $\widehat\pi$ is obtained from $\pi$ by removing $y$.
The vectors $\widehat q$ and $\widehat r$ contain the constant terms after the above
row scaling. 
Thus, the transformed system satisfies
$
\widehat P=-\widehat N^\top,
\qquad
\widehat Q=0,
$
and is an MLCP of the form
$
MLCP\left(
\widehat M,\widehat N,\widehat q,
-\widehat N^\top,0,\widehat r
\right).
$
\begin{assumption}[MLCP feasibility]
\label{ass:feas}
The linear constraints of the equivalent MLCP
\eqref{eq:micp_hat} are feasible; that is, there exist
$\widehat z\geq0$ and $\widehat\pi$ such that
\begin{equation}
\label{eq:mlcp_feas}
\widehat M\widehat z
+\widehat N\widehat\pi
+\widehat q
\geq0,
\qquad
-\widehat N^\top\widehat z+\widehat r=0.
\end{equation}
\end{assumption}
Assumption~\ref{ass:feas} requires only feasibility of the linear equalities and inequalities in \eqref{eq:mlcp_feas}, not the complementarity conditions, and therefore does not assume the existence of an equilibrium. In the present model, conventional electricity demand is price responsive and can adjust downward when generation or transmission capacity becomes scarce. The main potential source of infeasibility is therefore the fixed data center batch-processing requirement $q_b$ in \eqref{eq:kkt_N}, which must be satisfied regardless of market prices or conditions. The assumption can be verified by solving a linear feasibility problem.
\begin{theorem}[Existence of equilibrium]
\label{thm:exist}
Under Assumptions~\ref{ass:affine} and \ref{ass:feas}, and for
$\delta>0$, the MLCP \eqref{eq:micp} admits at least one solution
$(z^\star,\pi^\star)$.
\end{theorem}

\begin{proof}
By the MLCP existence result in \cite[Theorem 1.3.23]{Metzler}, a feasible
problem of the form in \eqref{eq:micp_hat}
admits a solution if $\widehat M$ is PSD.
Therefore, it remains to verify feasibility and positive
semidefiniteness of $\widehat M$.

Feasibility follows directly from Assumption~\ref{ass:feas}. Moreover,
by construction, all off-diagonal blocks of $\widehat M$ occur in
skew-symmetric pairs. Hence,
\begin{equation}
\label{eq:Mhat_sym}
\frac{\widehat M+\widehat M^\top}{2}
=
\operatorname{diag}\left(H_d,H_g,0,\ldots,0\right).
\end{equation}
Since $H_d$ and $H_g$ are PSD, so is $\widehat M$.  Thus, \eqref{eq:micp_hat} admits a solution.

Finally, since \eqref{eq:micp_hat} is equivalent to
\eqref{eq:micp}, a solution of the former induces a solution
$(z^\star,\pi^\star)$ of the latter.
\end{proof}

\subsection{MLCP Solution Uniqueness}
\label{sec:unique}
The preceding result establishes the existence of a market equilibrium, which need not be unique in general. Nevertheless, under linear inverse demand and marginal cost functions, the special MLCP structure introduced above allows us to establish uniqueness of certain aggregate quantities. 
\begin{theorem}[Uniqueness of aggregate quantities]
\label{thm:unique}
Under Assumptions~\ref{ass:affine} and \ref{ass:feas}, suppose
$\delta>0$. Further assume that 
$b^1_{it}>0$ in \eqref{eq:Marginal_Benefit} for all
$i\in I^d$, $t\in T$, and $c^1_{jh}>0$ in
\eqref{eq:Marginal_Cost} for all $j\in I$, $h\in H_j$.
Then, for any two market equilibria $(z^1,\pi^1)$ and
$(z^2,\pi^2)$, the aggregate demand $D_{it}$ defined in
\eqref{eq:aggD} and aggregate generator output $G_{jht}$ defined in
\eqref{eq:aggG} satisfy
\[
D^1_{it}=D^2_{it},
\qquad
\forall i\in I^d,\ t\in T,
\]
and
\[
G^1_{jht}=G^2_{jht},
\qquad
\forall j\in I,\ h\in H_j,\ t\in T.
\]
Hence, aggregate consumer demand and generator output are unique
across market equilibria.
\end{theorem}
\begin{proof}
Any equilibrium of \eqref{eq:micp} induces a solution of the
equivalent MLCP \eqref{eq:micp_hat}. By the MLCP uniqueness result
in \cite[Theorem 1.3.25]{Metzler}, for any two solutions
$(\widehat z^1,\widehat\pi^1)$ and
$(\widehat z^2,\widehat\pi^2)$ of a feasible MLCP of the form in \eqref{eq:micp_hat} 
with $\widehat M$ positive semidefinite,
\[
(\widehat M+\widehat M^\top)\widehat z^1
=
(\widehat M+\widehat M^\top)\widehat z^2.
\]
From \eqref{eq:Mhat_sym},
we have
$(\widehat M+\widehat M^\top)/2
=
\operatorname{diag}(H_d,H_g,0,\ldots,0),
$
and therefore
$
H_d d^1=H_d d^2$, and 
$H_g g^1=H_g g^2.
$

By the definitions of $H_d$ and $H_g$, these equalities imply
$
b^1_{it}(D^1_{it}-D^2_{it})=0
$
for each $(i,t)$ and
$
c^1_{jh}(G^1_{jht}-G^2_{jht})=0
$
for each $(j,h,t)$. Strict positivity of the corresponding slope
coefficients therefore yields uniqueness of $D_{it}$ and $G_{jht}$.
\end{proof}
\subsection{Contract Reshuffling and Additionality}
\label{subsec:reshuffling}

The uniqueness result above applies to aggregate physical quantities,
but does not generally determine bilateral contract allocations. To
characterize this remaining allocation freedom, define the aggregate
procurement of buyer $i$ at time $t$ as
$
G^{\mathrm B}_{it}
:=
\sum_{j\in I}\sum_{h\in H_j}g_{jhit}.
$ 
For fixed buyer-side procurement $G^{\mathrm B}$ and generator output
$G$, define
\begin{equation}
\label{eq:G_fiber}
\mathcal{G}(G^{\mathrm B},G)
:=
\left\{
\widetilde g\geq0:
\begin{aligned}
\sum_{j\in I}\sum_{h\in H_j}\widetilde g_{jhit}
&=G^{\mathrm B}_{it},\\
\sum_{i\in I}\widetilde g_{jhit}
&=G_{jht}
\end{aligned}
\right\},
\end{equation}
where the first equality holds for all $i\in I$, $t\in T$, and the
second for all $j\in I$, $h\in H_j$, $t\in T$.
Thus, $\mathcal{G}(G^{\mathrm B},G)$ is the set of bilateral
allocations consistent with the same buyer-side procurement and
generator output. Theorem~\ref{thm:unique} establishes uniqueness of
the generator output $G_{jht}$. For conventional consumers,
\eqref{eq:mk_1} implies $G^{\mathrm B}_{it}=D_{it}$, so their total
procurement is also unique. The theorem does not, however, determine
how these aggregate quantities are allocated across individual
buyer--generator contracts. We refer to variations
within $\mathcal{G}(G^{\mathrm B},G)$ as \textit{contract
reshuffling}.
\begin{proposition}[Nonuniqueness of bilateral allocations]
\label{prop:reshuffling}
Consider a market equilibrium and let $g$ denote its bilateral
generation allocation, with corresponding buyer-side procurement
$G^{\mathrm B}$ and generator output $G$. Fix $t\in T$. Suppose there
exist distinct buyers $i_1,i_2\in I$ and distinct generators
$(j_1,h_1)$ and $(j_2,h_2)$ such that
$
g_{j_1h_1i_2t}>0$ and 
$g_{j_2h_2i_1t}>0.
$
Then $\mathcal{G}(G^{\mathrm B},G)$ contains infinitely many
bilateral allocations. In particular, for any
\[
0<\varepsilon
\leq
\min\left\{
g_{j_1h_1i_2t},
g_{j_2h_2i_1t}
\right\},
\]
the allocation $\widetilde g=g+\Delta g$ belongs to
$\mathcal{G}(G^{\mathrm B},G)$, where
\begin{align}
\Delta g_{j_1h_1i_1t}&=\varepsilon,
&
\Delta g_{j_1h_1i_2t}&=-\varepsilon, \nonumber\\
\Delta g_{j_2h_2i_1t}&=-\varepsilon,
&
\Delta g_{j_2h_2i_2t}&=\varepsilon,
\label{eq:CR}
\end{align}
and all other components of $\Delta g$ are zero.
\end{proposition}
\begin{proof}
For each buyer, the changes in \eqref{eq:CR} sum to zero, and the
same holds for each generator. Hence $\widetilde g$ has the same
$G^{\mathrm B}$ and $G$ as $g$. The bound on $\varepsilon$ ensures
$\widetilde g\geq0$, so
$\widetilde g\in\mathcal{G}(G^{\mathrm B},G)$. Since
$\varepsilon$ can take infinitely many positive values within the
stated interval, $\mathcal{G}(G^{\mathrm B},G)$ contains infinitely
many elements.
\end{proof}

Moreover, all allocations in $\mathcal{G}(G^{\mathrm B},G)$ induce
the same physical power system outcome. Under \eqref{eq:mk_5}, nodal
net injections depend on the bilateral allocations only through
buyer-side procurement and generator output. Hence nodal injections
and transmission flows are invariant over
$\mathcal{G}(G^{\mathrm B},G)$. Generation costs and total system
emissions are also invariant because $G_{jht}$ is fixed.

Proposition~\ref{prop:reshuffling} therefore implies that contractual
carbon attribution need not be uniquely determined by the physical
power system outcome.

\begin{corollary}\rm{\textbf{(Nonuniqueness of contract-attributed emissions).}}
\label{cor:reshuffling}
Let $i_1\in I^\kappa$ be a hyperscaler and define its
contract-attributed emissions as
\begin{equation}
\label{eq:E_attr}
E^{\mathrm{attr}}_{i_1t}(\widetilde g)
:=
\sum_{j\in I}\sum_{h\in H_j}
e_{jh}\widetilde g_{jhi_1t}.
\end{equation}
Under the conditions of Proposition~\ref{prop:reshuffling}, if
$e_{j_1h_1}\neq e_{j_2h_2}$, then
$E^{\mathrm{attr}}_{i_1t}(\widetilde g)$ is nonconstant over
$\mathcal{G}(G^{\mathrm B},G)$. In particular, if
$e_{j_1h_1}<e_{j_2h_2}$, the reshuffling in \eqref{eq:CR} decreases
the hyperscaler's contract-attributed emissions while leaving the
physical power system outcome and total system emissions unchanged.
\end{corollary}

This result follows directly from Proposition~\ref{prop:reshuffling}.
Under \eqref{eq:CR}, the change in the hyperscaler's
contract-attributed emissions is
$
\varepsilon e_{j_1h_1}
-\varepsilon e_{j_2h_2}
=
-\varepsilon
\left(e_{j_2h_2}-e_{j_1h_1}\right),
$
which is nonzero whenever $e_{j_1h_1}\neq e_{j_2h_2}$ and negative
when $e_{j_1h_1}<e_{j_2h_2}$.

The contract reshuffling characterized here does not imply that an
executed PPA can be unilaterally reassigned. Rather, it captures the
allocation freedom that may exist when contracts are formed, renewed,
or otherwise adjustable. Existing contractual obligations restrict
this freedom and thus reduce the admissible subset of
$\mathcal{G}(G^{\mathrm B},G)$. The result therefore concerns the
contractual attribution of available generation rather than the
modification of fixed PPAs.

Corollary~\ref{cor:reshuffling} shows that a buyer can reduce its
contract-attributed emissions without changing the physical generation
outcome or total system emissions. This raises the question of when a
claim of cleaner electricity procurement should correspond to an actual
increase in clean generation. A common criterion for making this
distinction is \textit{additionality}. Here, additionality means that
an increase in a buyer's qualifying clean electricity procurement is
accompanied by an increase in qualifying clean generation, rather than
solely by a reallocation of existing clean generation among buyers.  This definition leads directly to the following implication for
contract reshuffling.
\begin{corollary}\rm{\textbf{(Additionality and contract reshuffling).}}
\label{cor:additionality}
Let $i_1\in I^\kappa$ denote a hyperscaler and let $\mathcal{R}$ denote
the set of qualifying clean generators. For any
$g,\widetilde g\in\mathcal{G}(G^{\mathrm B},G)$, define
$\Delta g:=\widetilde g-g$, the change in the hyperscaler's qualifying
clean procurement as
$
\Delta R_{i_1t}
:=
\sum_{(j,h)\in\mathcal{R}}
\Delta g_{jhi_1t},
$ 
and the corresponding change in aggregate qualifying clean generation as
$
\Delta G^{\mathcal{R}}_t
:=
\sum_{(j,h)\in\mathcal{R}}
\sum_{i\in I}\Delta g_{jhit}.
$
If additionality requires
$$
\Delta G^{\mathcal{R}}_t
\geq
\Delta R_{i_1t}>0,
$$
then such an increase in the hyperscaler's qualifying clean procurement
cannot be achieved through contract reshuffling within
$\mathcal{G}(G^{\mathrm B},G)$.
\end{corollary}
This result follows directly from the definition of
$\mathcal{G}(G^{\mathrm B},G)$. Since every allocation in this set
preserves the output of each generator,
$
\sum_{i\in I}\Delta g_{jhit}=0$, $\forall j\in I,\ h\in H_j,\ t\in T,
$
and hence $\Delta G^{\mathcal{R}}_t=0$. Therefore, the additionality
condition cannot hold when $\Delta R_{i_1t}>0$.

The above results motivate two questions for the numerical analysis:
whether lower contract-attributed emissions are accompanied by lower
physical system emissions, and whether restricting the flexibility to
reshuffle bilateral contracts narrows any gap between the two.
Section~\ref{sec:case} examines these questions numerically.
\section{Numerical Case Study} \label{sec:case}
\subsection{Data, Assumptions, and Scenarios}

We apply the IEEE Reliability Test System (RTS 24-Bus) \cite{rts99} to
illustrate the model described in Section~\ref{sec:model}. The system
consists of 24 buses, 38 transmission lines, 17 constant-power loads with
a total demand of 2,850 MW, and 32 generating units. Of the 32 generating
units, six hydropower units are excluded from the analysis because they
are assumed to operate at their maximum output of 50 MW \cite{wang03}.
We assume that all generators are owned by a single firm, as the analysis does not focus on market power.
To generalize the results, the transmission limit of each line is reduced to 60\% in order to produce non-uniform LMPs (Locational Marginal Prices). 
The marginal cost of generation is modeled as a linear function of output, parameterized by the coefficient vectors \( C^0 \) and \( C^1 \), whose elements \( c^0_{jh} \) and \( c^1_{jh} \) are defined in \eqref{eq:Marginal_Cost}.
We assume a hyperscaler is located at node 24, where three MDCs are situated in nodes 11, 12, and 17.
Node 24 is one of the nodes that experiences a lower power price at the baseline.\footnote{In fact, analysis in \cite{peskoe25} indicates that energy expenditures 
account for more than 40\% of total datacenter operating costs. 
Hyperscalers owned by major technology companies often collaborate with 
utilities to secure discounted long-term tariffs. Recent debates have 
centered on the extent to which other consumers bear a disproportionate 
share of the fixed costs induced by the rapid expansion of datacenters and the impact on local communities.
}
The parameter $\nu$ is assumed to be equal to 1,700 GPUs/MW, if hosted NVIDIA H100 or NVIDIA A100, which consumes roughly 700 W/GPU and 400 W/GUP per hour \cite{nvidia_dgxh100_user_guide}. 
For easy exposition, we simulate only a single period with $T=1$.



The RTS 24-Bus case is first solved as a least-cost minimization problem with fixed nodal load to obtain the dual variables. 
The dual variables, coupled with an assumed price elasticity of -0.2, are then used to define the affine inverse demand function. 
The estimated price elasticity of demand is consistent with findings from previous studies \cite{azevedo11}.

We examine the ability of a hyperscaler to pursue sustainability (i.e., lowering CO$_2$ emissions) by leasing computing infrastructure from MDCs. 
These MDCs are strategically co-located at the edge of the network, where renewable energy is frequently curtailed due to insufficient demand or limited transmission capacity.

We assume the capacities of MDC1--MDC3 are equal to 12, 10, and 2 MW, respectively. 
We assume that MDC1 can process all five batches, MDC2 can process only 
batches 1--3, and MDC3 can handle only batches 4 and 5.
The endowed ``curtailed'' renewables for three MDCs 1--3 equal 2, 4, and 5 MWh, respectively. 
The workload in MW that needs to be processed by the hyperscaler is equal to 75, 42, 51, 48 and 69, all in MW for types 1--5, respectively. All batches collectively account for roughly 10\% of the baseline loads.

Our primary results focus on two scenarios: 1) $\delta=0.1$ refers to the case when the hyperscaler places greater weight on  CO$_2$ emissions, and 2) $\delta=0.9$ refers to the case when the hyperscaler is concerned more about processing costs.
For each case, we investigate two emission disclosure schemes: {\it{ex ante}} and {\it{ex post}}. 
We also examine the impact of the hyperscaler's preference, denoted by $\delta$, on CO$_2$ emissions, processing costs, and leasing costs in Section \ref{sec:sen}. 

In reporting emissions, we distinguish contract-attributed data-center
emissions from total system emissions, with the latter determined by
aggregate generator output, consistent with the
distinction developed in Section~\ref{subsec:reshuffling}. 


\subsection{Ex Post Disclosure}
We first discuss the results of ``{\it{ex post}}'' emission disclosure in Table \ref{tab:cost_post}(a)--\ref{tab:rate_post}(a) for cost, emissions, and emission intensity, respectively, with a focus on two cases with $\delta=0.1$ (emissions) and $\delta=0.9$ (costs). 
Note that in both cases, the hyperscaler processes 74\% of the workload locally.

Table \ref{tab:cost_post}(a) suggests that when the hyperscaler places greater weight on emissions ($\delta=0.1$), the operator is willing to pay more (0.913 vs. 0.156 \$/GPU, reported in Table \ref{tab:rate_ante}(a) bottom panel) to have the workload processed by MDCs since they are powered by curtailed renewables. 
The total processing cost under $\delta=0.9$ is equal to \$71,983, which is significantly lower than that under $\delta=0.1$ at \$150,295. 

\begin{table}[htbp]
\centering
\caption{Results: Processing costs [\$] (top) \& Average Procurement cost[\$/MWh] (bottom)}
\vspace{0.2em}
\begin{minipage}[t]{0.47\linewidth}
\vspace{0pt}
\centering
\text{(a) {\it{Ex Post}}}\\
\setlength{\tabcolsep}{4pt}
\begin{tabular}{l r r}
\hline
\text{[\$]} & $\delta = 0.1$ & $\delta = 0.9$ \\
\hline
Local      & 113,031 & 65,617 \\
MDCs       & 37,264 & 6,366 \\
Total Cost & 150,295  & 71,983 \\
\hline
MDC1 & 58.89 & 58.89 \\
MDC2 & 52.91 & 52.91  \\
MDC3 & 0.00 & 0.00 \\
\hline
\end{tabular}
\label{tab:cost_post}
\end{minipage}
\hfill
\begin{minipage}[t]{0.47\linewidth}
\vspace{-3pt}
\centering
\text{(b) {\it{Ex Ante}}}\\[3pt]
\setlength{\tabcolsep}{4pt}
\begin{tabular}{l r r}
\hline
\text{[\$]} & $\delta = 0.1$ & $\delta = 0.9$ \\
\hline
Local      & 124,367 & 66,188 \\
MDCs       & 3,636  & 5,701 \\
Total Cost & 128,023 & 71,888\\
\hline
MDC1 & 61.06 & 58.57   \\
MDC2 &  55.44  & 52.31 \\
MDC3 & 0.00 & 0.00 \\
\hline
\end{tabular}
\label{tab:cost_ante}
\end{minipage}
\end{table}

Table~\ref{tab:em} summarizes the emissions from three sources: local
processing at the hyperscaler, outsourced processing at MDCs, and total
system emissions. 
To our surprise, when the hyperscaler focuses more on cost
(i.e., $\delta = 0.9$), its operation can actually also result in a lower level of emissions from outsourcing workloads to MDCs, while emissions from local processing at the hyperscaler remain nearly unchanged.
The workload-related CO$_2$ footprint under $\delta=0.9$ is reduced  by 3.56 t (= 39.59--43.15), or approximately 8\%.
This occurs because the {\it ex post} scheme is an imperfect instrument:
since emissions are disclosed only after operations, the hyperscaler cannot distinguish among MDCs by emission intensity when allocating workloads. 
Crucially, this reduction occurs in contract-attributed workload
emissions rather than in physical system emissions. Total system
emissions remain unchanged at 1{,}362.76 metric tons of CO$_2$ in the
two cases, illustrating Corollary~\ref{cor:reshuffling}: cleaner
contractual attribution need not correspond to lower physical system
emissions. We examine whether restricting this contract reshuffling can
mitigate the effect in Section~\ref{sec:joint}.


\begin{table}[htbp]
\centering
\caption{Results: Emissions [t]}
\label{tab:em}
\vspace{0.2em}
\begin{minipage}[t]{0.47\linewidth}
\centering
\setlength{\tabcolsep}{4pt}
\text{(a) {\it{Ex Post}}}\\[3pt]
\begin{tabular}{l r r}
\hline
 & $\delta = 0.1$ & $\delta = 0.9$ \\
\hline
Local        & 32.51  & 32.51 \\
MDCs         & 10.64  & 7.08 \\
Total        & 43.15 & 39.59 \\
\hline
System       &  1,362.76 & 1,362.76 \\
\hline
\end{tabular}
\label{tab:em_post}
\end{minipage}
\hfill
\begin{minipage}[t]{0.47\linewidth}
\centering
\text{(b) {\it{Ex Ante}}}\\[3pt]
\setlength{\tabcolsep}{4pt}
\begin{tabular}{l r r}
\hline
 & $\delta = 0.1$ & $\delta = 0.9$ \\
\hline
Local        & 34.50  & 32.51 \\
MDCs         & 1.40  & 6.46 \\
Total        & 35.90  & 38.97 \\
\hline
System       & 1,361.50 & 1,370.41\\
\hline
\end{tabular}
\label{tab:em_ante}
\end{minipage}
\end{table}

Table \ref{tab:rate_post}(a) also reports the emission intensity by MDCs. 
Except for MDC3, which has an intensity of 0 kg/MWh as it is powered only by renewables, the other two MDCs enter power purchase agreements with non-renewable suppliers.
The fact that emissions are disclosed {\it{ex post}} means that the hyperscaler does not have information to price GPU leasing differently based on emission intensity and thus leads to the same leasing offer for each MDC. 

\begin{table}[htbp]
\centering
\caption{Results: Emission rate [kg/MWh] (top) \& leasing cost [\$/GPU] (bottom)}
\vspace{0.2em}
\begin{minipage}[t]{0.47\linewidth}
\centering
\setlength{\tabcolsep}{4pt}
\text{(a) {\it{Ex Post}}}\\[3pt]
\begin{tabular}{l r r}
\hline
\text{} & $\delta = 0.1$ & $\delta = 0.9$ \\
\hline
MDC1 & 621.0 & 621.0 \\
MDC2 & 739.0 &  145.0\\
MDC3 & 0.0   & 0.0   \\
\hline
MDC1 & 0.913 & 0.156 \\
MDC2 & 0.913 & 0.156  \\
MDC3 & 0.913 & 0.156 \\
\hline
\end{tabular}
\label{tab:rate_post}
\end{minipage}
\hfill
\begin{minipage}[t]{0.47\linewidth}
\centering
\text{(b) {\it{{Ex Ante}}}}\\[3pt]
\setlength{\tabcolsep}{4pt}
\begin{tabular}{l r r}
\hline
\text{[kg/MWh]} & $\delta = 0.1$ & $\delta = 0.9$ \\
\hline
MDC1 & 621.1 & 559.1 \\
MDC2 & 621.0 & 145.0 \\
MDC3 & 0.00 & 0.00 \\
\hline
MDC1 & 0.036 & 0.127  \\
MDC2 &  0.033  & 0.152 \\
MDC3 & 0.930 & 0.157 \\
\hline
\end{tabular}
\label{tab:rate_ante}
\end{minipage}

\end{table}


\subsection{Ex Ante Disclosure}
The workload processing costs under the {\it ex ante} scheme are reported in Table~\ref{tab:cost_ante}(b). The hyperscaler maintains total costs at \$71{,}888 under $\delta = 0.9$, comparable to the {\it ex post} case (\$71{,}983). In contrast, when $\delta = 0.1$, overall workload costs decrease by 15\%, from \$150{,}295 in Table~\ref{tab:cost_post}(a) to \$128{,}023 in Table~\ref{tab:cost_ante}(b). Under {\it ex ante} disclosure, the hyperscaler observes MDC emission intensities and therefore values MDC capacity leasing less as $\delta$ increases and as emission intensity rises (see Section~\ref{sec:sen}). For example, when emission intensity is 0~t/MWh, Table~\ref{tab:rate_ante}(b) shows that the hyperscaler is willing to pay \$0.930/GPU for MDC3 at $\delta = 0.1$, compared to \$0.157/GPU at $\delta = 0.9$. A similar pattern appears under the {\it ex post} scheme, although interpretation is complicated because emission profiles are not distinguished. For a given $\delta$, lower MDC emission intensity leads to higher capacity leasing prices under {\it ex ante} disclosure. Consistent with this, MDC1 and MDC2 command higher leasing prices under $\delta = 0.9$ than under $\delta = 0.1$, for example \$0.127 versus \$0.036 for MDC1 and \$0.152 versus \$0.033 for MDC2. This differentiation does not arise under {\it ex post} disclosure. Despite these pricing effects, overall system emissions in Table~\ref{tab:em_ante}(b) remain comparable to those in Table~\ref{tab:em_ante}(a), indicating that {\it ex ante} disclosure alone does not yield meaningful sustainability gains.



\subsection{Forward Contract Cases} \label{sec:joint}
The preceding results show that contract reshuffling is possible when
bilateral electricity procurement can be freely reallocated across
buyers. In practice, however, conventional loads may already be served
under forward contracts that commit a portion of their purchases to
specific suppliers. Such contractual commitments restrict the
reallocation of existing generation and may therefore limit contract
reshuffling.

To examine this effect, we introduce a forward-contract restriction as
a numerical extension of the baseline equilibrium model. Let
$G^{\mathrm{FC}}_{jhit}$ denote the bilateral purchase from generator
$h$ at bus $j$ by conventional consumer $i$ in the baseline case
without datacenter loads. We impose the lower bound
\begin{equation}
g_{jhit}
\geq
\tau G^{\mathrm{FC}}_{jhit},
\qquad
\forall i\in I^d,\ j\in I,\ h\in H_j,\ t\in T,
\label{eq:producer_fc}
\end{equation}
where $\tau\in[0,1]$ denotes the forward-contract position. Thus,
$\tau$ specifies the minimum fraction of each conventional consumer's
baseline bilateral purchase that must be retained. A larger $\tau$
leaves less scope for reallocating generation among conventional
consumers and datacenters, thereby restricting contract reshuffling.
In the numerical experiments, we consider
$\tau\in\{0.6,0.7,0.8,0.9\}$.

Figure~\ref{fig:forward_cong}(a) plots total system CO$_2$ emissions
as a function of $\delta$ under these forward-contract requirements
for the {\it ex ante} disclosure scheme. The solid black line indicates the maximum system-wide emissions of
1{,}370 metric tons of CO$_2$. Emission reductions
occur primarily when the hyperscaler places greater weight on emissions, corresponding to $\delta\leq0.6$. When the forward-contract position reaches
90\%, leaving substantially less flexibility for contract reshuffling,
system emissions are approximately 2--5\% below this benchmark.
These results illustrate that restricting the reallocation of
conventional consumers' existing contracts can translate the
hyperscaler's preference for cleaner procurement into reductions in
physical system emissions.

Figure~\ref{fig:forward_cong}(b) reports the corresponding network
congestion costs. Under the 60\%--80\% forward-contract cases,
congestion changes relatively little because substantial procurement
flexibility remains. Under the 90\% case, however, congestion declines
as $\delta$ increases. As the hyperscaler becomes more cost-aware, it
allocates more workload to MDCs to reduce processing costs, which in
this case also alleviates network congestion.

\subsection{Sensitivity Analysis} \label{sec:sen}
Figure~\ref{fig:alpha_co2rate}(a) shows the leasing costs of the three MDCs as a function of $\delta$. 
The solid blue denotes the {\it{ex post}} scheme while the color-coded dashed lines represent the {\it{ex ante}} scheme for three MDCs. 
Under the {\it{ex post}} scheme, the hyperscaler is only aware of CO$_2$ after the fact.
Thus, its willingness to pay for outsourcing workload is the same for all three MDCs.
However, when the hyperscaler becomes aware of the respective emission impact of each MDC under the {\it{ex ante}} scheme before the fact, as alluded to earlier, the operator's willingness to pay for the services then also depends on the emission intensity of each MDC in Fig. \ref{fig:alpha_co2rate}(b). 
Given MDC3 accrues no emission, the capacity leasing cost curve of MDC3 lies above the other two MDCs until the $\delta$ approaches 1 where the processing  cost is the only concern.

\begin{figure}[!t]
\centering
\begin{minipage}{0.49\columnwidth}
    \centering
       (a)
        \includegraphics[width=\linewidth]{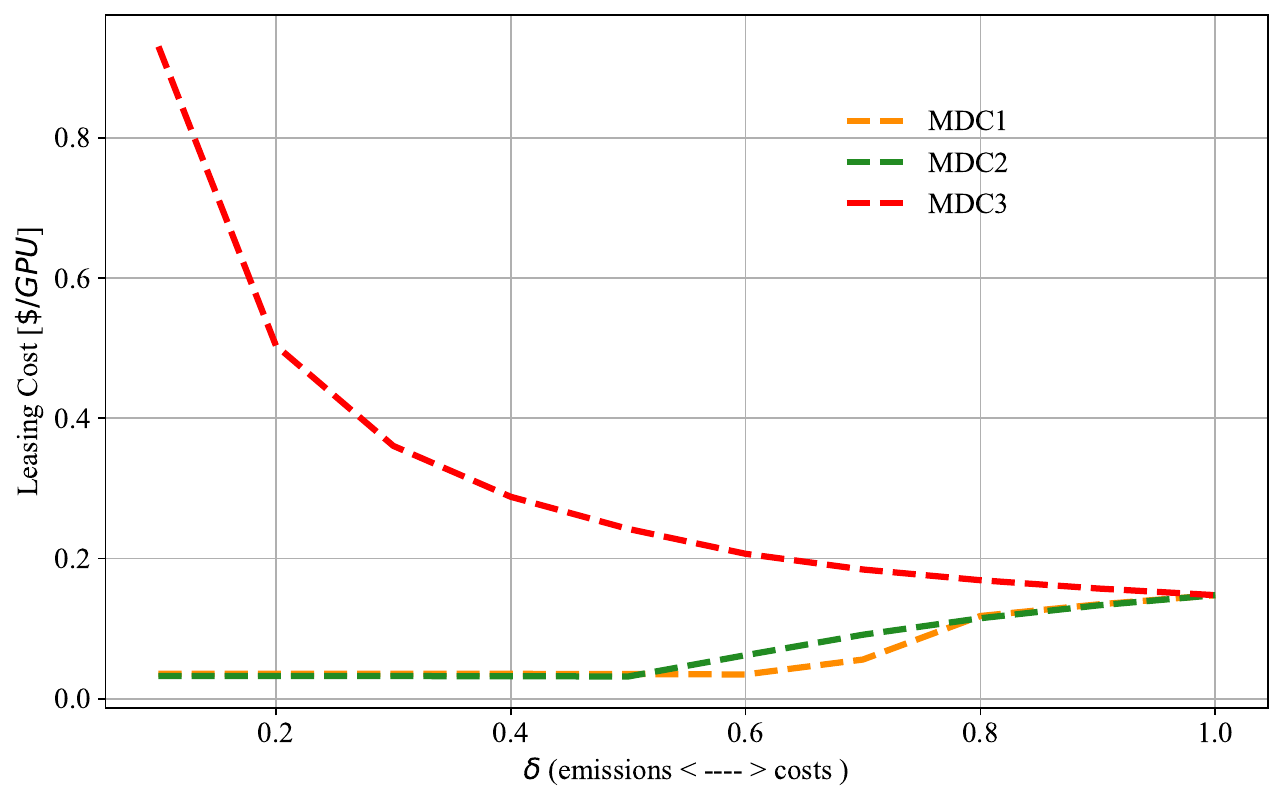}
\end{minipage}
\hfill
\begin{minipage}{0.49\columnwidth}
    \centering
   (b)  
    \includegraphics[width=\linewidth]{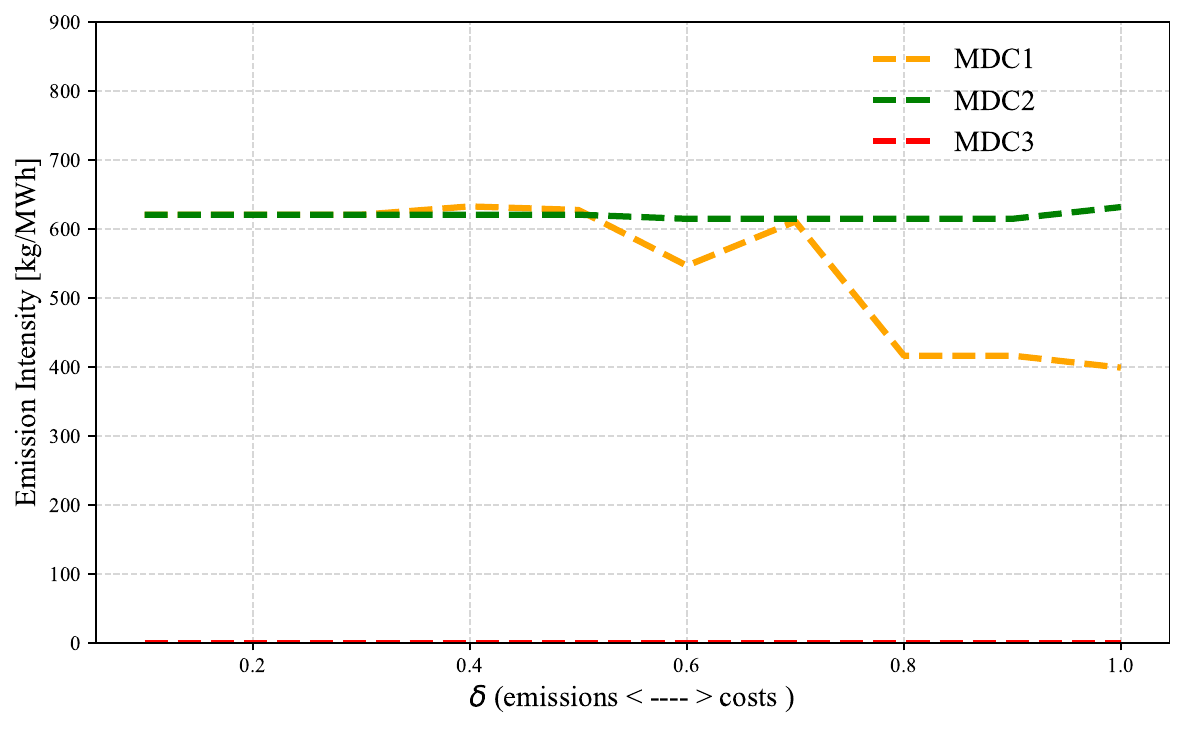}
\end{minipage}
\caption{(a) MDC leasing cost $\alpha$ and (b) MDC CO$_2$ emissions rate against $\delta$}
\label{fig:alpha_co2rate}
\end{figure}

However, the comparison between MDC1 and MDC2 is less clear. 
Because MDC2 is able to procure power at a lower price by an average of 6\$/MWh  across the scenarios, the owner is willing to operate the facility at a lower capacity leasing price (for $\delta$ between 0.5 and 0.7) while still maintaining its profit margin even though its emission intensity is lower as alluded to in Fig. \ref{fig:alpha_co2rate}(b). 
Overall, as $\delta$ approaches 1, emission intensity is no longer a concern for the hyperscaler, eventually leading to identical leasing prices for all MDCs.


Figures~\ref{fig:emission}–\ref{fig:cost} present datacenter workload-related processing costs and CO$_2$ emissions for the {\it ex post} and {\it ex ante} disclosure cases in panels (a) and (b), respectively. Each figure shows three components: workloads processed locally by the hyperscaler (red dashed line), workloads outsourced to MDCs (green dashed line), and total quantities (blue dashed line). In both disclosure regimes, increasing $\delta$ shifts the datacenter toward cost minimization, resulting in lower local and total processing costs, as shown in Fig.~\ref{fig:cost}. At the same time, CO$_2$ emissions rise and eventually spike as $\delta$ approaches 1, as shown in Fig.~\ref{fig:emission}, when processing cost becomes the dominant objective. Collectively, these figures illustrate a trade-off between CO$_2$ emissions and total processing costs. Across all cases under both schemes, however, total system CO$_2$ emissions remain essentially unchanged, indicating that contract reshuffling among conventional consumers, MDCs, and the hyperscaler yields no meaningful emission reduction. We examine the role of forward contracts in the next section.

\begin{figure}[!t]
\centering
\begin{minipage}{0.49\columnwidth}
    \centering
        {(a) \it{ex post}}
        \includegraphics[width=\linewidth]{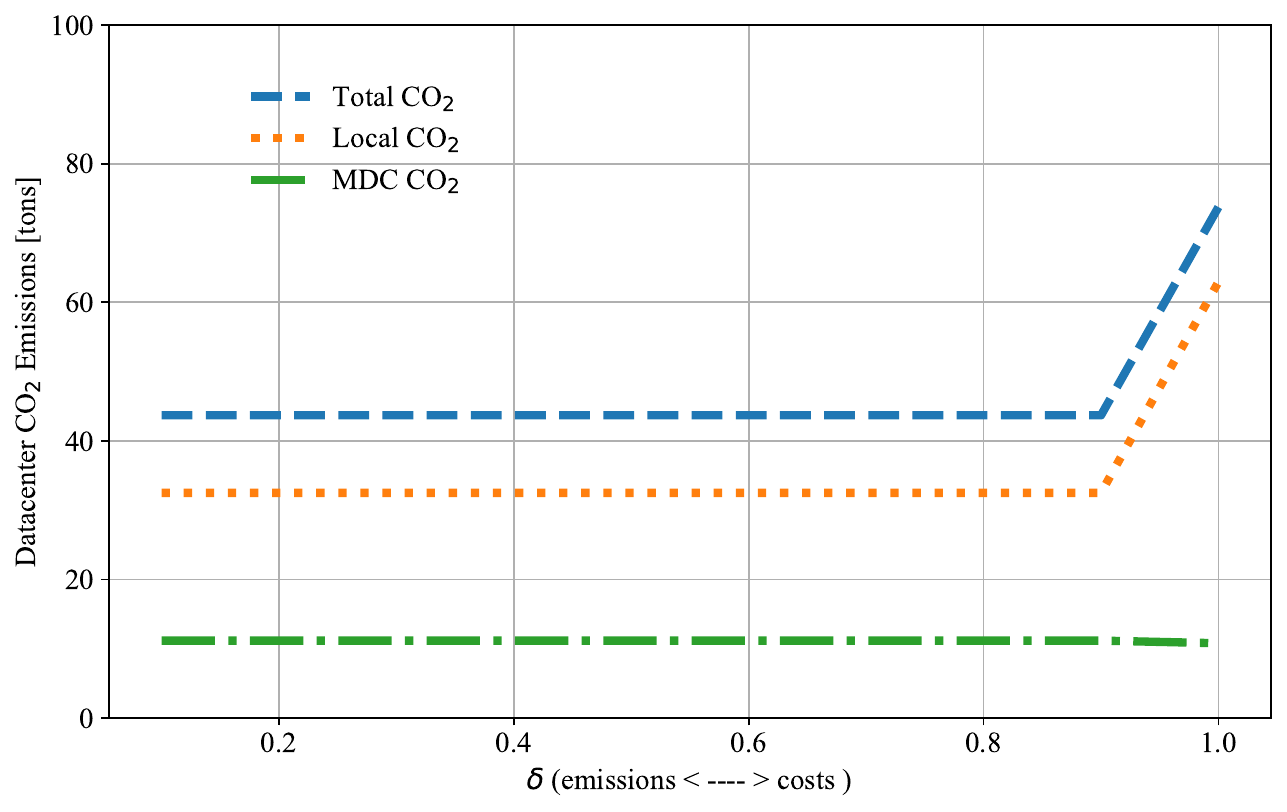}
\end{minipage}
\hfill
\begin{minipage}{0.49\columnwidth}
    \centering
    {(b) \it{ex ante}}    \includegraphics[width=\linewidth]{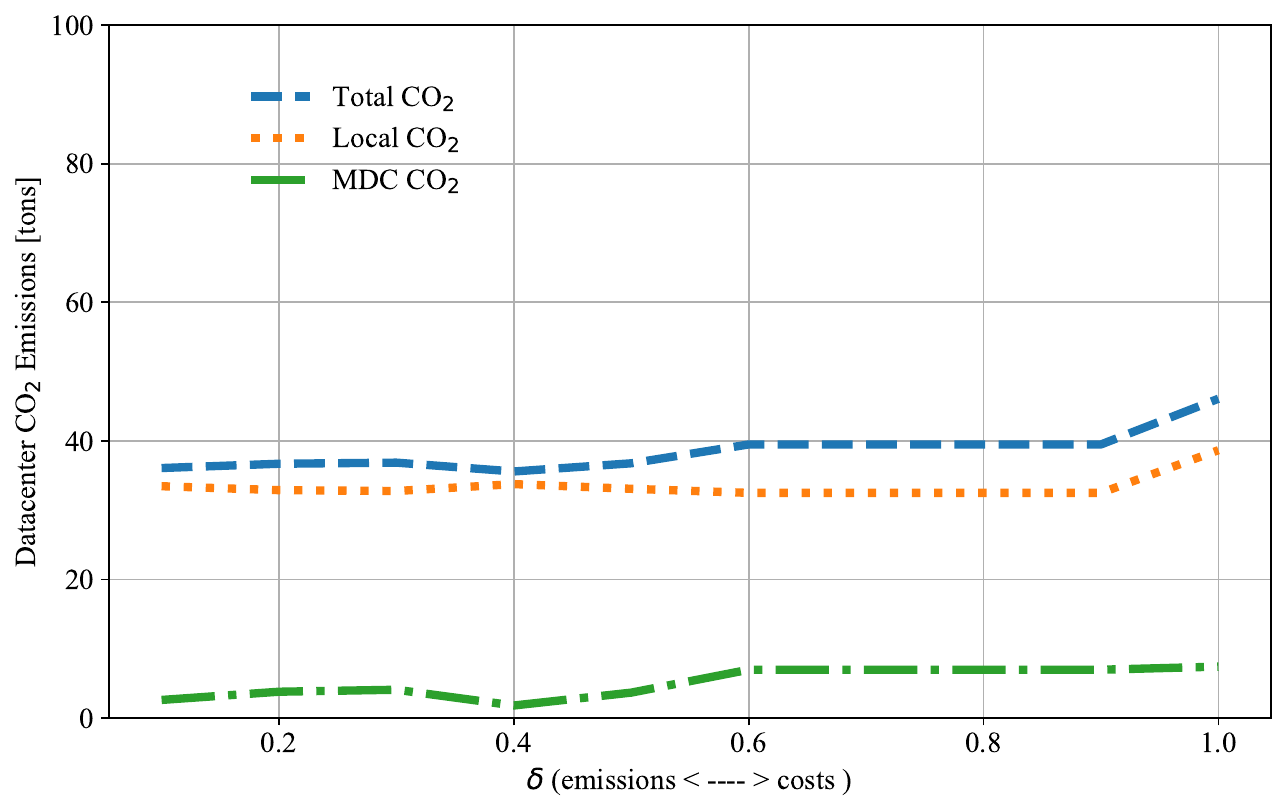}
\end{minipage}
\caption{datacenters workload CO$_2$ emissions against $\delta$}
\label{fig:emission}
\end{figure}

\begin{figure}[!t]
\centering
\begin{minipage}{0.49\columnwidth}
    \centering
   {(a) \it{ex post}} 
   \includegraphics[width=\linewidth]{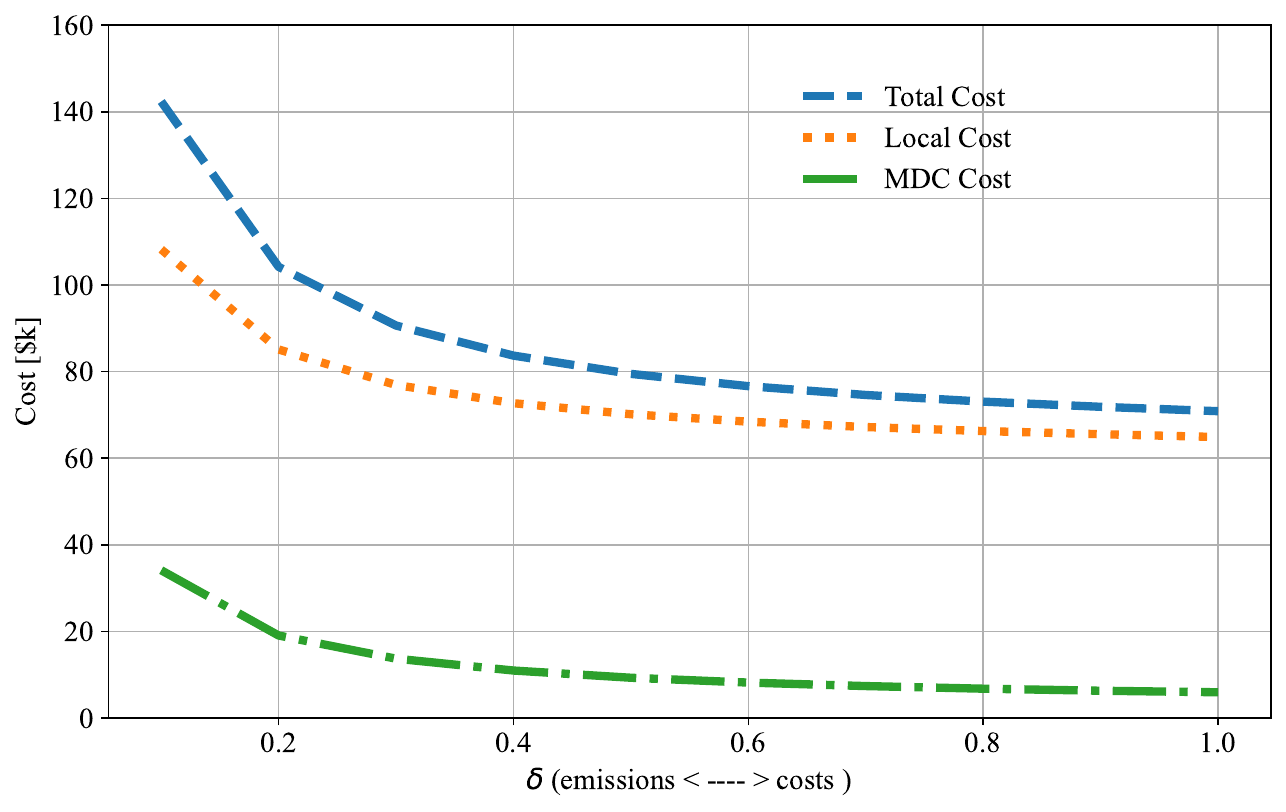}
\end{minipage}
\hfill
\begin{minipage}{0.49\columnwidth}
    \centering
    {(b) {\it{ex ante}}} 
    \includegraphics[width=\linewidth]{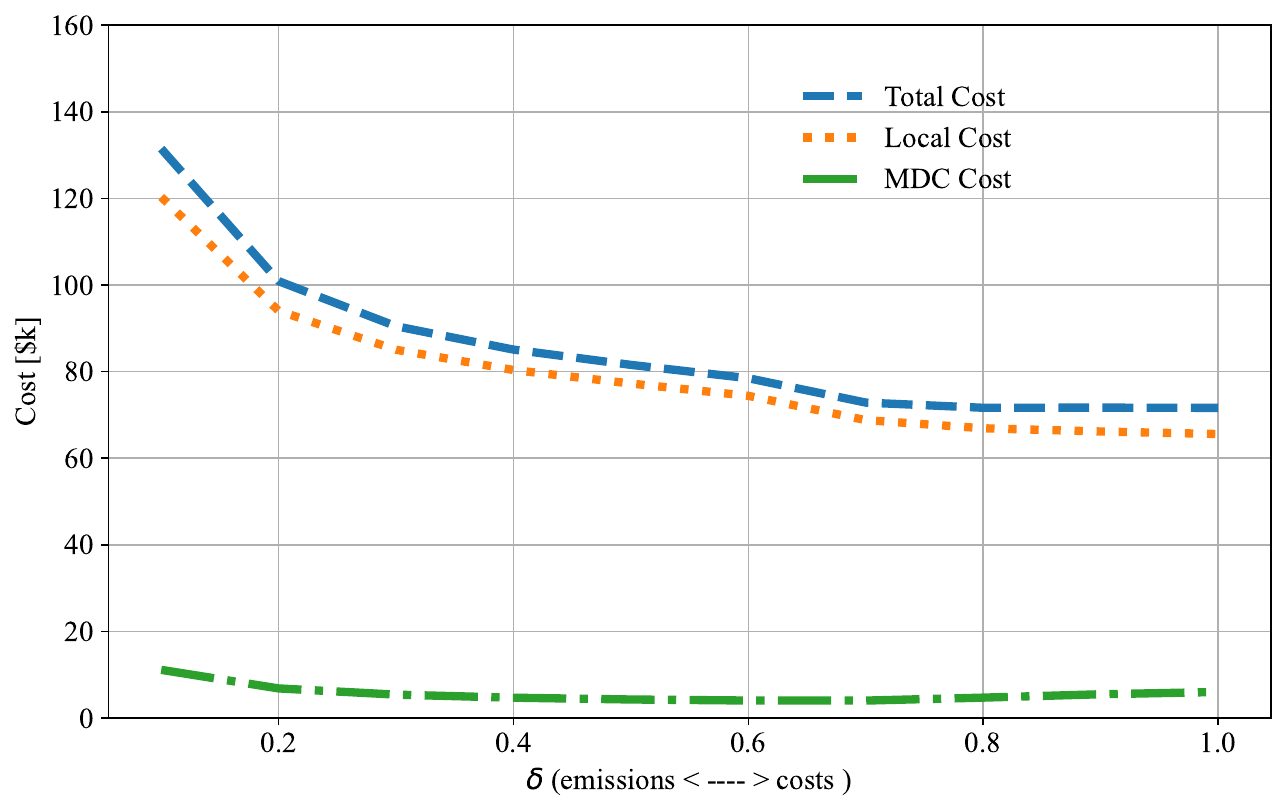}   
\end{minipage}
\caption{datacenters workload processing costs against $\delta$}
\label{fig:cost}
\end{figure}


\begin{figure}[!t]
\centering
\begin{minipage}{0.49\columnwidth}
    \centering
    {(a) CO$_2$ emission} \includegraphics[width=\linewidth]{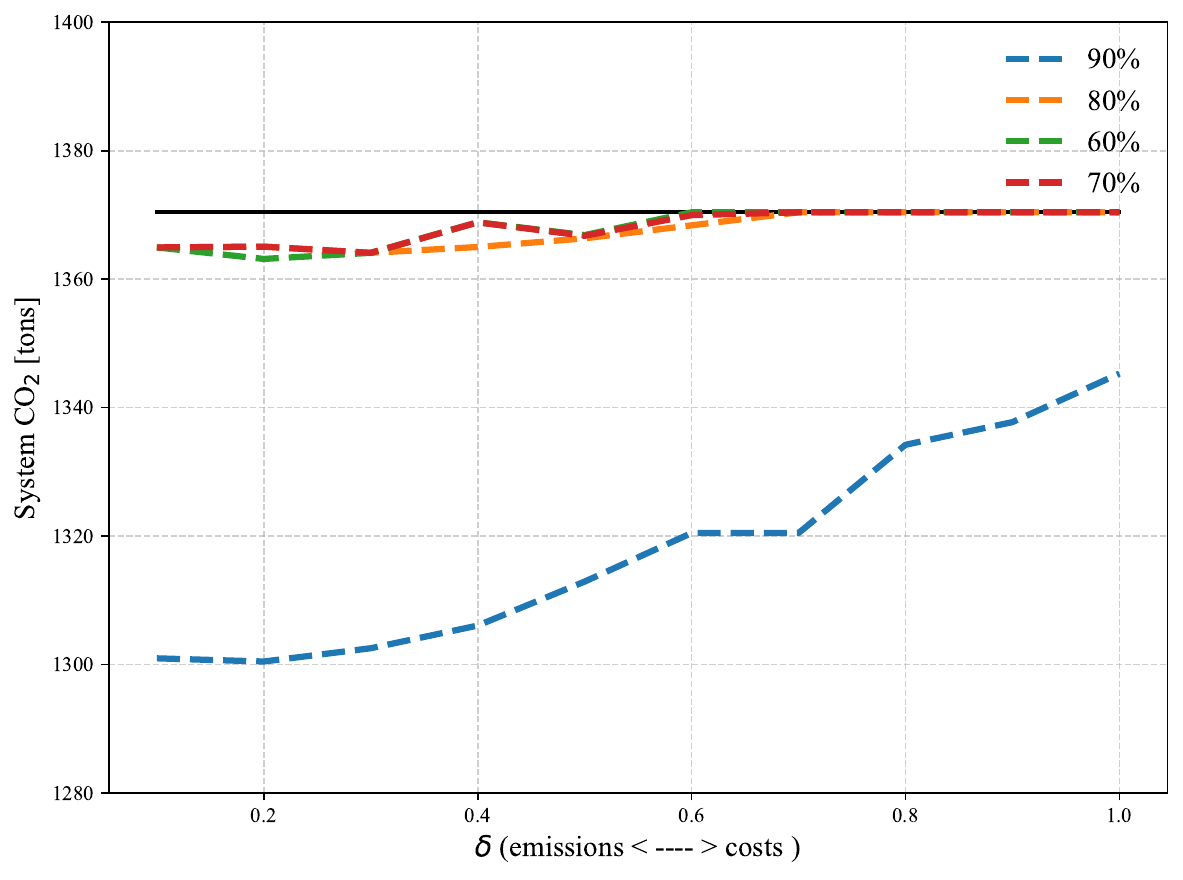}
\end{minipage}
\hfill
\begin{minipage}{0.49\columnwidth}
    \centering
    {(b) Congestion} \includegraphics[width=\linewidth]{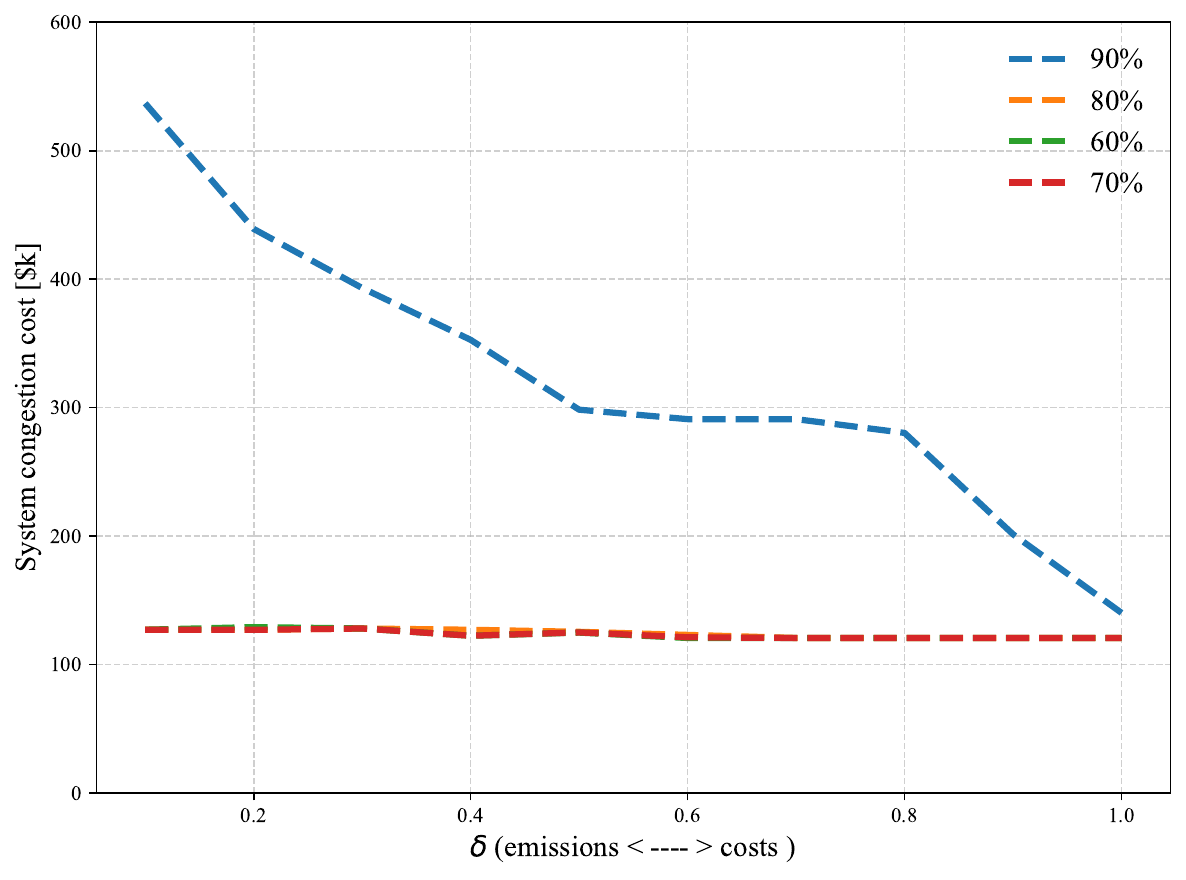}   
\end{minipage}
\caption{Total system CO$_2$ emissions and congestion cost against $\delta$ under different forward contract positions for the {\it{ex ante}} scheme}
\label{fig:forward_cong}
\end{figure}


\section{Conclusions}\label{sec:concl}

This paper studies whether carbon-aware migration of inference workloads
to geographically distributed MDCs can translate into reductions in
physical power-system emissions. We develop a market-equilibrium model
that jointly represents capacity leasing between a hyperscaler and MDCs,
bilateral electricity procurement, and network-constrained power-system
operations. Two emission-disclosure schemes are considered. Under the
{\it ex post} scheme, MDC emissions are disclosed after operation,
whereas under the {\it ex ante} scheme, MDC emission intensities are
available to the hyperscaler when it makes its workload-allocation and
capacity-leasing decisions.

Our theoretical analysis establishes existence of a market equilibrium
and uniqueness of aggregate consumer demand and generator output under
the stated conditions. More importantly, these unique physical quantities
need not uniquely determine bilateral electricity allocations. Electricity
purchases can be reshuffled across buyers while preserving generator
output, allowing a hyperscaler to obtain a cleaner contractual attribution
without a corresponding reduction in physical system emissions. This
distinction between contractual attribution and physical outcomes is
central to evaluating the emissions implications of carbon-aware
computing.

The IEEE RTS-24 case study illustrates the practical significance of
this mechanism. Under {\it ex ante} disclosure, carbon information
affects the hyperscaler's capacity-leasing offers and can substantially
reduce the emissions attributed to its workloads, while the corresponding
change in total system emissions remains small. When conventional
consumers retain larger portions of their existing bilateral purchases
through forward contracts, however, the scope for contract reshuffling
is reduced and carbon-aware workload allocation produces larger physical
emission reductions. A notable network-level effect also emerges under a
sufficiently restrictive forward-contract position: as the hyperscaler
becomes more cost-aware, it shifts additional workload to MDCs at the
network edge, thereby reducing congestion in the case study. These results
suggest that carbon disclosure alone may be insufficient to ensure that
cleaner procurement translates into cleaner physical generation; the
contractual structure of electricity procurement also matters.

Our analysis is subject to several limitations. First, we assume that all
workload batches yield the same revenue regardless of their latency
requirements. In practice, inference requests can differ substantially in
their time sensitivity and economic value. Accounting for heterogeneous
service requirements could therefore affect both workload allocation and
the capacity-leasing prices offered to individual MDCs. Second, we assume
that the amount of renewable energy that would otherwise be curtailed is
known and fixed. In practice, renewable curtailment is uncertain and varies
with system conditions; consequently, capacity-leasing decisions may depend on
expectations about future renewable availability. Extending the framework
to incorporate such uncertainty provides an important direction for
future research.

\appendix

\section{Data}
Table \ref{tab:data_summary} provides the generator data used in the case study. Load data of conventional consumers are directly based on RTS 24 system \cite{rts99}, and all the other data, e.g., PTDF, will be available upon request. 

\begin{table}[htbp]
  \centering
  \caption{Generator Data}
  \label{tab:data_summary}
  \begin{tabular}{crrrr}
    \toprule
    $h$ & \multicolumn{1}{c}{$c_0$} & \multicolumn{1}{c}{$c_1$} & \multicolumn{1}{c}{$e$} & \multicolumn{1}{c}{$P_{\max}$} \\
    \midrule
    1  & 48.4 & 0.44 & 633 & 20  \\
    2  & 48.4 & 0.44 & 633 & 20  \\
    3  & 11.0 & 0.01 & 145 & 76  \\
    4  & 11.0 & 0.01 & 145 & 76  \\
    5  & 48.4 & 0.44 & 633 & 20  \\
    6  & 48.4 & 0.44 & 633 & 20  \\
    7  & 11.0 & 0.01 & 145 & 76  \\
    8  & 11.0 & 0.01 & 145 & 76  \\
    9  & 25.4 & 0.07 & 615 & 100 \\
    10 & 25.4 & 0.07 & 615 & 100 \\
    11 & 25.4 & 0.07 & 615 & 100 \\
    12 & 28.5 & 0.02 & 739 & 197 \\
    13 & 28.5 & 0.02 & 739 & 197 \\
    14 & 28.5 & 0.02 & 739 & 197 \\
    15 & 38.9 & 0.08 &  56 & 12  \\
    16 & 38.9 & 0.08 &  56 & 12  \\
    17 & 38.9 & 0.08 &  56 & 12  \\
    18 & 38.9 & 0.08 &  56 & 12  \\
    19 & 38.9 & 0.08 &  56 & 12  \\
    20 &  9.3 & 0.01 & 220 & 155 \\
    21 &  9.3 & 0.01 & 220 & 155 \\
    22 & 13.5 & 0.00 & 621 & 400 \\
    23 & 13.5 & 0.00 & 621 & 400 \\
    24 &  9.3 & 0.01 & 220 & 155 \\
    25 &  9.3 & 0.01 & 220 & 155 \\
    26 &  8.6 & 0.01 & 440 & 350 \\
    \bottomrule
  \end{tabular}
\end{table}

\section{The MLCP Formulation -- Additional Details}
\label{sec:AppA}
This appendix provides additional details on the matrices $N$, $P$,
and $Q$, and the vectors $q$ and $r$ in
\eqref{eq:micp_comp}--\eqref{eq:micp_eq}. 
The matrix $N$ collects the coefficients of the free variables $\pi$
in the complementarity conditions. Consistent with the ordering of
$z$ and $\pi$ in \eqref{eq:z_def} and \eqref{eq:pi_def}, let
$N_{u,v}$ denote the block of $N$ corresponding to the $u$-block of
$z$ and the $v$-block of $\pi$. The nonzero blocks of $N$ are
\begin{align}
N_{d,\theta^d} &= E_d, &
N_{g,\theta^d} &= -E_g^d, \nonumber \\
N_{g,\theta^\chi} &= -E_g^\chi, &
N_{g,\theta^\kappa} &= -E_g^\kappa, \nonumber \\
N_{g,\omega} &= A_\omega, &
N_{\mu^1,y} &= A_{\mu y}, \nonumber\\
N_{\mu^2,y} &= -A_{\mu y}, &
N_{p,\theta^\chi} &= E_p^\chi, \nonumber\\
N_{p,\eta} &= E_{p\eta}, &
N_{k^r,\eta} &= -E_{kr\eta}, \nonumber\\
N_{k^r,\alpha} &= -\nu E_{kr\alpha}, &
N_{s,\eta} &= -E_{s\eta}, \nonumber\\
N_{k^s,\alpha} &= \delta\nu E_{ks\alpha}, &
N_{k^s,\psi} &= E_{ks\psi}, \nonumber\\
N_{\ell,\theta^\kappa} &= \delta E_{\ell\theta}, &
N_{\ell,\psi} &= E_{\ell\psi}.
\label{eq:N_blocks}
\end{align}
All remaining blocks of $N$ are zero. Here, each $E_\bullet$ is an appropriately dimensioned selection or
incidence matrix with entries in $\{0,1\}$. The matrix $A_\omega$
maps the wheeling charges to the producer stationarity conditions so
that the row corresponding to $g_{jhit}$ contains
$\omega_{it}-\omega_{jt}$, while $A_{\mu y}$ is defined such that
$
(A_{\mu y}y)_{kt}
=
\sum_{i\in I}\mathrm{PTDF}_{ki}y_{it}.
$

The matrices $P$ and $Q$ collect, respectively, the coefficients of
the nonnegative variables $z$ and the free variables $\pi$ in the
linear equality conditions
$
Pz+Q\pi=r.
$
Consistent with the ordering of $\pi$ in \eqref{eq:pi_def}, these
equalities are obtained by stacking the grid operator conditions
\eqref{eq:kkt_yfree} and \eqref{eq:sumy_0}, the market-clearing
conditions \eqref{eq:mk_1}--\eqref{eq:mk_alpha}, the MDC
energy-balance condition \eqref{eq:mdc_energy_balance}, and the
hyperscaler batch-processing condition \eqref{eq:kkt_N}. The matrix
$P$ contains the coefficients multiplying $z$, following the sign
conventions of the corresponding equality conditions, while $Q$
contains the coefficients multiplying $\pi$. The constant right-hand sides of these equality conditions are collected
in $r$. Its only nonzero blocks are
\[
\hspace*{40pt}
r_\eta=\vec g^{\,c},
\qquad
r_\psi=\vec q_B.
\]
Here, each arrow denotes the vector obtained by stacking and, where
necessary, replicating the corresponding parameter to match the
dimension of the associated block of $\pi$. In particular,
$\vec g^{\,c}$ stacks the total curtailed renewable energy
$\sum_{h\in H_i^\omega}g^c_{iht}$ over $i\in I^\chi$ and $t\in T$,
whereas $\vec q_B$ stacks $q_b$ over $b\in B$ and replicates it over
$t\in T$. All remaining components of $r$ are zero.

Finally, the vector $q$ collects the constant terms in the complementarity
conditions. Its nonzero blocks are
\begin{equation*}
\begin{aligned}
q_d&=-\vec b^{\,0},&
q_g&=\vec c^{\,0},&
q_\lambda&=\vec G,\\
q_{\mu^1}&=\vec F,&
q_{\mu^2}&=\vec F,&
q_\rho&=\vec{\mathrm{Cap}},\\
q_\upsilon&=\vec g^{\,c},&
q_\ell&=(1-\delta)\vec e.
\end{aligned}
\end{equation*}

\bibliographystyle{cas-model2-names}
\bibliography{datacenter}

\end{document}